\documentclass[%
prx, reprint,
superscriptaddress,
 amsmath,amssymb,
 aps,
]{revtex4-1}

\usepackage[english]{babel}
\usepackage[utf8]{inputenc}

\setcitestyle{sort&compress,numbers}
\usepackage{amsthm, color, mathtools, amsmath, amssymb, setspace,bbm, tensor, subfigure, mathtools, placeins, graphicx, wrapfig,float, verbatim, xcolor}
\usepackage{graphicx}
\usepackage[unicode=true, breaklinks=false, pdfborder={0 0 1}, backref=false, colorlinks=true, linkcolor=blue, citecolor=blue]{hyperref}

\makeatletter
\g@addto@macro{\UrlBreaks}{\UrlOrds}
\makeatother

\usepackage{cancel,bm}

\let\oldsqrt\sqrt
\def\sqrt{\mathpalette\DHLhksqrt}
\def\DHLhksqrt#1#2{%
\setbox0=\hbox{$#1\oldsqrt{#2\,}$}\dimen0=\ht0
\advance\dimen0-0.2\ht0
\setbox2=\hbox{\vrule height\ht0 depth -\dimen0}%
{\box0\lower0.4pt\box2}}

\newcommand{\tr}{\operatorname{tr}}

\newcommand{\iu}{{i\mkern1mu}}

\newcommand*\xbar[1]{%
   \hbox{%
     \vbox{%
       \hrule height 0.5pt 
       \kern0.5ex
       \hbox{%
         \kern-0.2em
         \ensuremath{#1}%
         \kern-0.0em
       }%
     }%
   }%
} 

\newcommand{\Bcal}{\mathcal{B}}
\newcommand{\Ccal}{\mathcal{C}}

\newcommand{\Hcal}{\mathcal{H}}
\newcommand{\Ical}{\mathcal{I}}
\newcommand{\Mcal}{\mathcal{M}}

\newcommand{\Ucal}{\mathcal{U}}

\newcommand{\Jcal}{\mathcal{J}}

\newcommand{\Pprob}{\mathbb{P}}

\newcommand{\ident}{\mathbbm{1}}

\newcommand{\eq}{ Eq.}

\newcommand{\Ads}{\mathbbm{A}}
\newcommand{\Bds}{\mathbbm{B}}
\newcommand{\Xds}{\mathbbm{X}}

\def\BraVert{\egroup\,\mid\,\bgroup}

\def\ketbra#1#2{\ket{#1\vphantom{#2}}\!\bra{#2\vphantom{#1}}}

\def\bra#1{\mathinner{\langle{#1}|}}

\def\ket#1{\mathinner{|{#1}\rangle}}

\def\braket#1{\mathinner{\langle{#1}\rangle}}

\usepackage{bbm, braket}
\usepackage[mathscr]{euscript}
\allowdisplaybreaks
\newtheorem*{theorem*}{Thm.}

\newtheorem{observation}{Observation}

\newtheorem*{conjecture*}{Conjecture}
\newtheorem*{corollary*}{Corollary}

\makeatletter
\newcommand{\neutralize}[1]{\expandafter\let\csname c@#1\endcsname\count@}
\makeatother

\newtheorem{thm}{Thm.}

\begin{document}

\title{Quantum chicken-egg dilemmas: Delayed-choice causal order and non-classical correlations}

\author{Simon Milz}
\email{simon.milz@oeaw.ac.at} 
\affiliation{Institute for Quantum Optics and Quantum Information, Austrian Academy of Sciences, Boltzmanngasse 3, 1090 Vienna, Austria}
\affiliation{School of Physics and Astronomy, Monash University, Clayton, Victoria 3800, Australia}

\author{Dominic Jurkschat}
\affiliation{School of Physics and Astronomy, Monash University, Clayton, Victoria 3800, Australia}

\author{Felix A. Pollock}
\affiliation{School of Physics and Astronomy, Monash University, Clayton, Victoria 3800, Australia}

\author{Kavan Modi}
\affiliation{School of Physics and Astronomy, Monash University, Clayton, Victoria 3800, Australia}

\begin{abstract}
Recent frameworks describing quantum mechanics in the absence of a global causal order admit the existence of causally indefinite processes, where it is impossible to ascribe causal order for events $A$ and $B$. These frameworks even allow for processes that violate the so-called causal inequalities, which are analogous to Bell's inequalities. However, the physicality of these exotic processes is, in the general case, still under debate, bringing into question their foundational relevance. While it is known that causally indefinite processes can be probabilistically realised by means of a quantum circuit, along with an additional conditioning event $C$, concrete insights into the ontological meaning of such implementation schemes have heretofore been limited. Here, we show that causally indefinite processes can be realised with schemes where $C$ serves only as a classical flag heralding which causally indefinite process was realised. We then show that there are processes where \textit{any} pure conditioning measurement of $C$ leads to a causally indefinite process for $A$ and $B$, thus establishing causal indefiniteness as a basis-\textit{independent} quantity. Finally, we demonstrate that quantum mechanics allows for phenomena where $C$ can \textit{deterministically} decide whether $A$ comes before $B$ or vice versa, without signalling to either. This is akin to Wheeler's famous delayed-choice experiment establishing definite causal order in quantum mechanics as instrument-\textit{dependent} property. 
\end{abstract}

\date{\today}
\maketitle

\section{Introduction}
\label{sec::Intro}

Genuine quantum properties, like entanglement and coherence play an important role in many protocols and current or near future technologies~\cite{preskill_quantum_2018}. While these \textit{spatial} properties of quantum systems, and their resourcefulness have been studied in depth, much less is known about their temporal counterparts. Recent research has begun investigating the structure of temporal correlations of quantum systems~\cite{hoffmann_structure_2018, mao_structure_2020} as well as the quantification of quantum resources required to simulate temporal correlations~\cite{spee_simulating_2020}. While this program is in its early stages, the foundational importance of temporal (quantum) correlation is becoming clear. For instance, it has been demonstrated that temporal quantum correlations can enhance the performance of ticking clocks~\cite{budroni_nonclassical_2020}. The counterpart to no-signalling conditions, which play a crucial role in studies of spatial correlations, are conditions imposing causality. However, even when subject to these conditions, quantum mechanics yields surprises; within the field of quantum causal modelling~\cite{1367-2630-18-6-063032, allen_quantum_2017}, it has been shown that quantum mechanics allows for the superposition of common-cause and direct-cause causal structures~\cite{feix_quantum_2017, maclean_quantum-coherent_2017} as well as the violation of instrumental tests~\cite{chaves_quantum_2018} -- two feats that are not possible within the realm of classical causal models. Additionally, quantum mechanics can provide a speed-up in the discovery of causal relations~\cite{ried_quantum_2015, chiribella_quantum_2019}.

This is just the tip of the quantum iceberg; processes that are \textit{causally ordered} form only a subset of those allowed by quantum theory. The possibility to coherently control causal orders has drawn considerable recent interest, both on the theoretical~\cite{chiribella_quantum_2013,ebler_enhanced_2018,chiribella_indefinite_2018,salek_quantum_2018, procopio_communication_2019, zych_bells_2019, chiribella_quantum_2019a, wilson_quantum_2020,guerin_communication_2019,abbott_communication_2020, wilson_diagrammatic_2020, mukhopadhyay_superposition_2020, kristjansson_resource_2020}, as well the experimental~\cite{procopio_experimental_2015,rubino_experimental_2017,goswami_indefinite_2018, wei_experimental_2019, goswami_increasing_2020, guo_experimental_2020, taddei_experimental_2020, rubino_experimental_2020} side, and such control has been shown to be a resource in information theoretic tasks~\cite{feix_quantum_2015,guerin_exponential_2016,ebler_enhanced_2018,taddei_quantum_2019}. Going further, Ref.~\cite{OreshkovETAL2012} showed the existence of processes that are locally causal, but do not have a global causal order. Moreover, there it was shown that such processes allow for richer communication tasks than those with global causal order.

Specifically, the authors of Ref.~\cite{OreshkovETAL2012} constructed a so-called \textit{causal inequality}, which is reminiscent of Bell's inequalities and showed that quantum mechanics allows for processes that violate them, i.e., outperform causally ordered processes (classical, quantum, or beyond) in information theoretic games~\cite{OreshkovETAL2012, branciard_simplest_2015}. Further stratifying the structure of such causally indefinite processes, it has been demonstrated that there are causally non-separable processes, i.e., processes that cannot be represented as a convex mixture of causally ordered ones that do not violate causal inequalities~\cite{araujo_witnessing_2015, feix_causally_2016, araujo_purification_2017} and thus are reminiscent of entangled states that do not violate Bell's inequalities~\cite{werner_quantum_1989}. On the other hand, it has been shown that, beyond the two-party case, there are fully classical processes that violate causal inequalities~\cite{baumeler_maximal_2014, baumeler_perfect_2014a}. 

While such exotic causal structures are not \textit{a priori} prohibited by fundamental laws of physics~\footnote{See Ref.~\cite{araujo_purification_2017} for an investigation of processes under the requirement of \textit{purifiability}.}, their physicality, along with their implications, remains uncertain. In addition, and in stark contrast to otherwise spatially analogous entanglement, it is generally not clear how to experimentally implement causally indefinite processes \textit{deterministically}. However, \textit{probabilistic} protocols for realising an arbitrary process by means of a quantum circuit, i.e., a causally ordered process, with conditioning have been proposed~\cite{1367-2630-18-7-073037, silva_connecting_2017, araujo_quantum_2017, milz_entanglement_2018}, and the interconversion between properties of the employed circuit and the conditional causally indefinite process has been investigated~\cite{milz_entanglement_2018}.

Remarkably, as we discuss in this paper, no quantum correlations are required to realise causally indefinite processes via conditioning. Specifically, within the probabilistic implementation scheme of Ref.~\cite{milz_entanglement_2018}, for any process, there exists a quantum circuit which only displays classical correlations between the conditioning degrees of freedom and the remaining degrees of freedom of interest. This absence of quantum correlations allows for the interpretation that each measurement outcome on the conditioning system merely reveals -- but does not create -- the causally non-separable processes that was `realised' in the individual run, and establishes causality as a principle that holds on average, but not necessarily for individual runs of an experiment.

\begin{figure}
    \centering
    \includegraphics[width=0.78\linewidth]{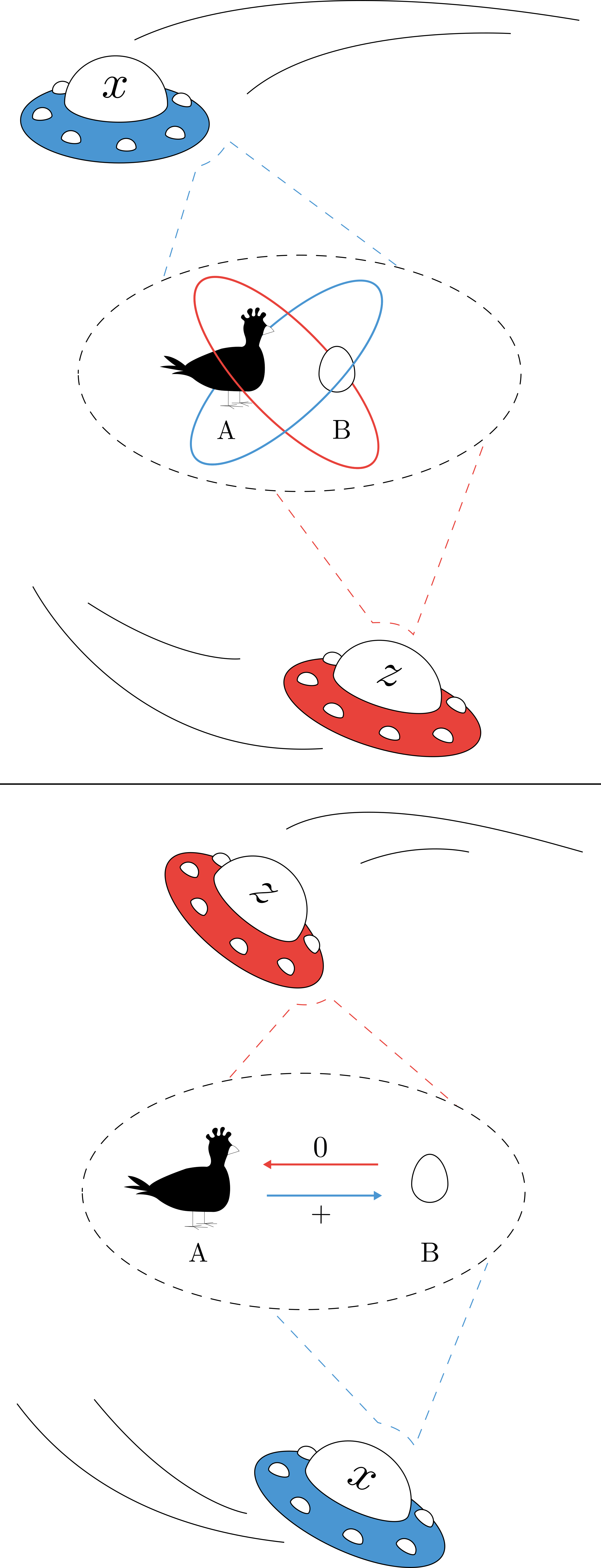}
    \caption{\textbf{Basis-independent causal disorder (top).} The conditioning basis of an observer depends on their orientation with respect to a fixed direction. For example, in either of the depicted cases the observer might condition in their respective Pauli-$z$ basis, yet with respect to a fixed reference, these measurements amount to measurements in the Pauli-$x$ (depicted in blue) and Pauli-$z$ (depicted in red) basis, respectively. As we show in Sec.~\ref{sec::CausalCoherence} there are processes for which the causal non-separability of the observed processes is independent of the conditioning basis. \textbf{Basis-dependent causal order (bottom).} Depending on the basis the observer measures in, the resulting conditional process is of order $A\prec B$ (for the blue case) or of order $B\prec A$ (for the red case). See Sec.~\ref{sec::CondCausOrder} for details. For simplicity, the conditioning system is omitted in the figure.}
    \label{fig::Comic}
\end{figure}

While this latter interpretation has the obvious objection that the causal ordering of an individual run of an experiment is not a meaningful notion \textit{per se}, it nonetheless raises the question of whether entanglement between the conditioning degrees of freedom and the rest is possible and/or enhances the conditioning scheme. This question is in the spirit of those regarding the resource that is used in the aforementioned studies of coherent control of causal orders; there, it is the entanglement between the relevant degrees of freedom and a control qubit that is crucial for all observed advantages (losing this qubit destroys the respective enhancements). In a similar vein, we show here that such coherent control can be used to make the conditioning procedures inherently `stable'. Specifically, the properties of the conditional processes crucially depend on the choice of measurement basis that is employed for the conditioning; we demonstrate that the range of conditioning bases that lead to causally non-separable processes can be vastly increased when entanglement is added, and that there are indeed causally ordered processes that lead to a causally non-separable process for \textit{any} conditioning basis. Such causally ordered processes, then, make causal non-separability an effect that stems from conditioning in a highly fine-tuned manner (as conditioning on most causally ordered processes will not yield a proper process~\cite{milz_entanglement_2018}), but renders it a property that is independent of how the conditioning apparatus is oriented with respect to the laboratories of Alice and Bob. Consequently, we shall call this property \textit{basis-independent} (see Fig.~\ref{fig::Comic} for a graphical representation).

Our first set of results establishes a connection between correlations and the properties of processes, and shows that entanglement can enhance conditioning scenarios, but is not a prerequisite for the realisation of causally indefinite processes. Our next result is even more surprising: We find physical processes where the conditioning party can choose the causal direction between two events, despite occurring after them. That is, we demonstrate that causal order itself can be understood as a basis-\textit{dependent} property; if the conditioning measurements are made in one basis, then $A$ occurs before $B$, but if they are made in another, then $B$ occurs before $A$. Importantly, as we show by explicit example, this basis-dependence occurs \textit{deterministically}; the respective basis choice fixes which of the opposing causal orders the processes had. In contrast to other scenarios considered in the literature~\cite{castro-ruiz_dynamics_2018}, here, the operation that determines the causal order between Alice and Bob happens \textit{after} their operations are performed, thus allowing one to choose causal order after the fact, instead of predetermining it.

This is akin to the famous delayed-choice experiment~\cite{wheeler, peruzzo, PhysRevLett.120.190401}, and we emphasise that this contextual behaviour is genuinely quantum and -- as we show -- cannot exist in the classical world. Put less prosaically, in quantum mechanics, the chicken-egg dilemma fundamentally has no resolution -- even when the underlying process is causally ordered -- but one's conclusion depends on how one `looks' at the process at hand (see Fig.~\ref{fig::Comic} for a graphical representation). On the other hand, as we show, with only one conditioning measurement (with two outcomes), it is \textit{not} possible to obtain processes of opposing causal order, mirroring similar results in the unconditional case, where -- in many simple cases -- it is impossible to superpose opposing causal orders~\cite{yokojima_consequences_2020, costa_no-go_2020}.

Before presenting these results, we begin by introducing the process matrix formalism, which is designed to represent spatio-temporal processes, including the those that do have a definite causal order.

\section{Process Matrix Preliminaries}

\subsection{General framework}
Throughout this article, we focus on two parties, Alice ($A$) and Bob ($B$), who perform generalized measurements in their distinct laboratories. We are interested in the joint probabilities they can possibly obtain when each of them employs an instrument $\Jcal_X = \{\Mcal_X^{(k)}\}_k$, with $X\in\{A,B\}$. An instrument is a collection of completely positive (CP) maps $\Mcal_X^{(k)}$, each describing the transformation on the observed system corresponding to one of a possible set of measurement outcomes. Moreover, the CP maps add up to a CP trace preserving (CPTP) map $\Mcal_X = \sum_k \Mcal_X^{(k)}$. Each of the CP maps $\Mcal_X^{(k)}$ transforms the quantum states from an input space $\Bcal(\Hcal_{X_{I}})$ to an output space $\Bcal(\Hcal_{X_O})$, i.e., $\Mcal_X^{(k)}: \Bcal(\Hcal_{X_I}) \rightarrow \Bcal(\Hcal_{X_O})$, where $\Hcal_{X_{I/O}}$ are the respective system Hilbert spaces, and $\Bcal(\Hcal_{X_{I/O}})$ denotes the set of matrices on said Hilbert space. Throughout, the dimension of the involved Hilbert spaces is considered to be finite and $d_X$ is the dimension of $\Hcal_X$.

For ease of notation, we employ the Choi-Jamio\l kowski isomorphism~\cite{jamiolkowski_linear_1972, Choi1975} to express all objects we consider as positive matrices. With this, every CP map $\Mcal_X^{(k)}:\Bcal(\Hcal_{X_I}) \rightarrow \Bcal(\Hcal_{X_O})$ corresponds to a positive matrix $M_X^{(k)}\in \Bcal(\Hcal_{X_O}\otimes \Hcal_{X_I})$, and every CPTP map $\Mcal_X$ corresponds to a positive matrix that additionally satisfies $\tr_{X_O}M_X = \ident_{X_I}$, where $\ident_{X_I}$ is the identity matrix on $\Hcal_{X_I}$.

In such a setting, owing to the linearity of quantum mechanics (in the sense that its statistics have to satisfy linearity of mixing), the joint probability for Alice and Bob to obtain outcomes $i$ and $j$, given that they used instruments $\Jcal_A$ and $\Jcal_B$, can then be computed via an equation of the form
\begin{gather}
\label{eqn::FirstW}
    \Pprob(i,j|\Jcal_A,\Jcal_B) = \tr[W(M_A^{(i)} \otimes M_B^{(j)})]\, ,
\end{gather}
where $W \in \Bcal(\Hcal_{A_0} \otimes \Hcal_{A_I} \otimes \Hcal_{B_O} \otimes \Hcal_{B_I})$ is called the \textit{process matrix}~\cite{OreshkovETAL2012} that encapsulates the spatio-temporal relations between $A$ and $B$. It accounts for the cases where Alice and Bob are causally connected, e.g., where Alice's operations can influence Bob's. In addition, it also captures the case where their causal order is indefinite.

Consequently, Eq.~\eqref{eqn::FirstW} has been dubbed the Born rule for temporal processes~\cite{chiribella_memory_2008, shrapnel_updating_2017}. The process matrix and its action are graphically depicted in Fig.~\ref{fig::Conditioning_Circuit}. Importantly, it contains all spatio-temporal correlations that are present between Alice and Bob. For example, as mentioned, $W$ can describe all conceivable scenarios where Alice's operations come before Bob's (denoted by $W^{A\prec B}$), Bob's operations come before Alice's (denoted by $W^{B\prec A}$), as well as situations, where Alice and Bob are spacelike separated (denoted by $W^{A\|B}$).

Following the literature, we will often call process matrices that display a definite causal order \textit{quantum combs}, or just combs~\cite{chiribella_quantum_2008,chiribella_theoretical_2009}. Any process matrix $W$ that can be represented as a probabilistic mixture of causally ordered processes, i.e., 
\begin{gather}
\label{eqn::CausSep}
    W = q W^{A\prec B} + (1-q) W^{B\prec A}\, ,
\end{gather}
is called \textit{causally separable}~\cite{OreshkovETAL2012}. The case $W^{A\|B}$ can be understood as a special case of $W^{B\prec A}$ or $W^{A\prec B}$ in\eq~\eqref{eqn::CausSep}. Here, causal order implies that a later choice of instrument cannot influence statistics at an earlier point in time. It has been shown~\cite{chiribella_quantum_2008,chiribella_theoretical_2009} that, for the two-party case we consider, this requirement implies 
\begin{gather}
\begin{split}\label{eqn::Causality_constraints}
    W^{X\prec Y}_{X Y} &= \ident_{Y_O} \otimes W^{X\prec Y}_{X Y_I} \\
    \text{and} \ \ \tr_{Y_I} W^{X\prec Y}_{XY_I} &= \ident_{X_O} \otimes \rho_{X_I}\, ,
\end{split}
\end{gather}
where $\rho_{X_I}$ is a quantum state, and we have added subscripts to signify which spaces the respective elements are defined on. For compactness, we will often employ the convention $X_IX_O := X$ when denoting spaces by subscripts. 

As a process $W^{A\|B}$ both satisfies $A\prec B$ \textit{and} $B\prec A$, the above conditions imply that 
\begin{gather}
    W^{A\|B} = \ident_{A_OB_O} \otimes \rho_{A_IB_I}\, . 
\end{gather}
Naturally, independently of what CPTP map Alice (Bob) performs, the `remaining' comb on Bob's (Alice's) side has to be causally ordered. We will call this property \textit{local} causality. Importantly, requiring that $W$ does not violate \textit{local} causality (in each of the respective laboratories A and B) does \textit{not} force it to abide by a fixed \textit{global} causal order (nor a convex combination of fixed causal orders)~\cite{oreshkov_causal_2016}. Specifically, local causality imposes the constraint 
\begin{gather}
\label{eqn::Def_W}
    \tr[W(M_{A} \otimes M_{B})] = 1 \quad \forall \ \text{CPTP maps} \ M_{A}, M_{B}, 
\end{gather}
and there exist process matrices, dubbed \textit{causally non-separable}, that satisfy\eq~\eqref{eqn::Def_W} but which cannot be represented as a probabilistic mixture of the form of\eq~\eqref{eqn::CausSep}. Additionally, there are process matrices that can violate \textit{causal inequalities}~\cite{OreshkovETAL2012,branciard_simplest_2015}; i.e., their causal indefiniteness can be verified in a device-independent way. It has been shown that not every causally non-separable process matrix can violate a causal inequality~\cite{araujo_witnessing_2015, feix_causally_2016, araujo_purification_2017}, implying the existence of causally non-separable process matrices that admit a causal model. This is analogous to the spatial setting, where there are entangled states that cannot violate any Bell inequality, and which admit a hidden variable model~\cite{werner_quantum_1989}. In what follows, we will also call processes that lack a clear causal order -- either in the weaker sense of causal non-separability, or in the stronger sense that they can violate a causal inequality -- \textit{causally indefinite}.

\subsection{Process matrices via conditioning}
Processes with a fixed causal order can always be understood as coming from a quantum circuit with a pure initial state and unitary dynamics~\cite{chiribella_theoretical_2009}. Causally separable processes, then, can be seen as a convex mixture of such circuits, e.g., beginning with a coin flip that decides which of the circuits is run. However, there is no such circuit dilation for causally non-separable processes~\footnote{While there are experimental implementations of the quantum switch~\cite{procopio_experimental_2015, goswami_indefinite_2018, wei_experimental_2019,procopio_communication_2019, taddei_experimental_2020}, it does not possess a representation in terms of a circuit where each of the laboratories occurs only once~\cite{chiribella_quantum_2013}.}.

On the other hand, it has been shown that \textit{any} process matrix~\footnote{In general, however, not all matrices obtained via conditioning will automatically satisfy the conditions required for a process matrix. Thus, not all conditional matrices are proper process matrices.}, causally non-separable or not, can be realised by means of a causally ordered process with an additional conditioning~\cite{chiribella_theoretical_2009, 1367-2630-18-7-073037,silva_connecting_2017, araujo_quantum_2017, milz_entanglement_2018}. To this end, we now introduce the third cast member of this paper, Charlie ($C$), who will be responsible for the conditioning. For example, the ordering of the overall process could be taken to be $A\prec B \prec C$, where the conditioning occurs in Charlie's laboratory (corresponding to a measurement of the degrees of freedom denoted by $C_I$). Then, for every process matrix $W \in \Bcal(\Hcal_A\otimes \Hcal_B)$, there exists a causally ordered $\Upsilon^{A\prec B\prec C} \in \Bcal(\Hcal_{A}\otimes \Hcal_{B} \otimes \Hcal_{C_I})$ such that
\begin{gather}
    \begin{split}
     \label{eqn::Equality_Prob}
    &\Pprob(i,j|\Jcal_A,\Jcal_B) = \tr[W(M_A^{(i)} \otimes M_B^{(j)})] \\
    &= \tfrac{1}{p_C(0)} \tr[\Upsilon^{A\prec B \prec C} (M_A^{(i)} \otimes M_B^{(j)} \otimes \ketbra{0}{0}_{C_I})]
    \end{split}
\end{gather}
holds for all $\{M_A^{(i)}, M_B^{(j)}\}$, where $p_C(0)$ is the probability to obtain outcome $0$ when measuring the system $C_I$ in the computational basis. 

We emphasize that, in principle, \textit{every} positive matrix $W$, proper process matrix or not, could be `realised' in the above way. However, proper processes, i.e., positive matrices $W$ that satisfy Eq.~\eqref{eqn::Def_W} are singled out in the sense that they form the largest set of positive matrices for which the success probability $p_C(0)$ is independent of the instruments $\Jcal_A$ and $\Jcal_B$~\cite{silva_connecting_2017, milz_entanglement_2018}. In this sense, proper processes are the only ones for which the above conditioning rule is linear~\cite{silva_connecting_2017} and the conditioning procedure is  fully independent of Alice's and Bob's actions. As this fails to hold for matrices that are not proper processes, their realisation via conditioning is somewhat ill-defined (or at least non-linear, as they require a renormalization that depends on what instruments Alice and Bob employ). Here, and in what follows, we will denote the comb corresponding to the overall circuit by $\Upsilon \in \Bcal(\Hcal_A\otimes \Hcal_B \otimes \Hcal_{C_I})$ to distinguish it from the realised process matrices (denoted by $W\in\Bcal(\Hcal_A\otimes \Hcal_B)$).

In line with the aforementioned causality requirements, a causally ordered process matrix as the one employed above, satisfies
\begin{gather}
\label{eqn::Causality3}
    \tr_{C_I}\Upsilon^{X\prec Y\prec C}_{XYC_I} = \ident_{Y_0} \otimes W^{X\prec Y}_{XY_I}\,,
\end{gather}
where $W^{X\prec Y}_{XY_I}$ obeys the causality constraints~\eqref{eqn::Causality_constraints}. Unsurprisingly then, the resulting process matrix on $XY$ is causally ordered if \textit{no} conditioning takes place on $C_I$ (i.e., the degrees of freedom $C_I$ are traced out). Put differently, denoting the process matrix obtained from conditioning on the outcome $i$ on $C_I$ by $W^{(i)}$, we see that $\sum_{i}p_C(i)W^{(i)}$ is causally ordered. Consequently, being in possession of the system $C_I$ is a crucial control resource for realising causally non-separable process matrices. 

In Ref.~\cite{milz_entanglement_2018}, an overall circuit -- shown in Fig.~\ref{fig::Conditioning_Circuit}, and henceforth referred to as `parallel' -- for the realisation of arbitrary process matrices, requiring two initial maximally entangled states, a qubit ancillary degree of freedom and a five-partite unitary, was provided. As Alice and Bob cannot signal to each other in this circuit, while Charlie comes after both of them, in the above convention, its causal order is of the form $A \| B\prec C_I$.

Following the notation of Fig.~\ref{fig::Conditioning_Circuit}, for every process matrix $W\in\Bcal(\Hcal_A \otimes \Hcal_B)$, there exists a unitary map $\Ucal$ acting on $A_O,A_I', C_I', B_I',B_O$, such that 
\begin{gather}
\begin{split}
    &\tr[W(M_A^{(i)} \otimes M_B^{(j)})]\\
    &= \tfrac{1}{p_C(0)} \tr \{\ketbra{0}{0}_{C_I} \Ucal\circ \Mcal_A^{(i)} \otimes \Mcal_B^{(j)} [\rho_{\Ads_I\Bds_IC_I'}]\} \\
    &= \tfrac{1}{p_C(0)}\tr[\Upsilon^{A\prec B \prec C_I}(M_A^{(i)} \otimes M_B^{(j)} \otimes \ketbra{0}{0}_{C_I})]\, ,
    \end{split}
\end{gather}
where $\rho_{\Ads_I\Bds_IC_I'} = \Phi^+_{\Ads_I} \otimes \Phi^+_{\Bds_I} \otimes \ketbra{0}{0}_{C_I}$, $\Xds_{I/O} = X_{I/O}X_{I/O}'$, and we have omitted identity maps and matrices where they appear. Evidently, since Alice and Bob cannot influence each other in this scenario, the overall process when discarding the qubit $C_I$ is of the type $W^{A\|B}$, and it is easy to see that 
\begin{gather}
    p_C(0) W^{(0)} + p_C(1) W^{(1)}  = \tfrac{1}{d_{A_I}d_{B_I}} \ident_{AB}\, ,
\end{gather}
where, as before, $W^{(0)}$ and $W^{(1)}$ are the process matrices obtained for the two different measurement outcomes on $C_I$. We stress that the success probability for this circuit is \textit{always} non-vanishing, and given by $p_C(0) = 1/(d_{A_I}d_{B_I} \lambda_\mathrm{max})$, where $d_X$ is the dimension of $\Hcal_X$ and $\lambda_\mathrm{max}$ is the maximal eigenvalue of the realised process matrix. 

\begin{figure}[ht!]
    \centering
    \includegraphics[width=0.6\linewidth]{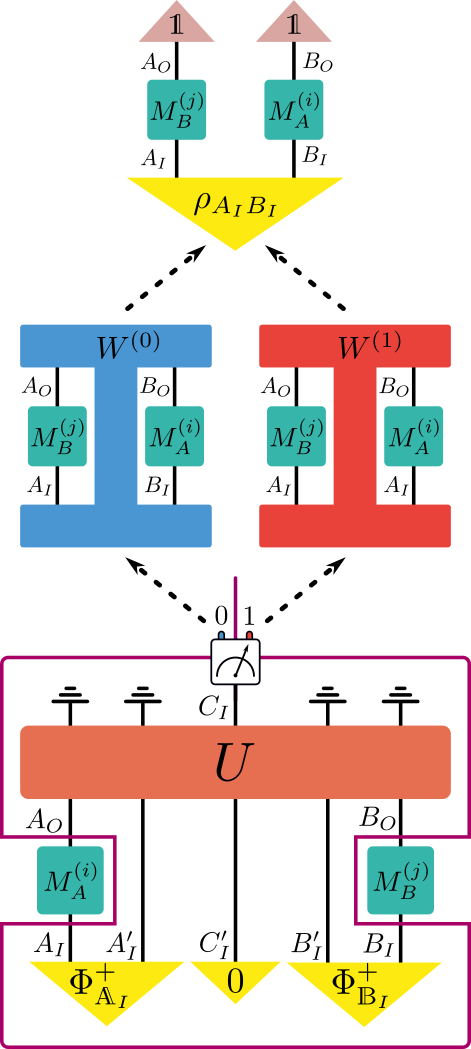}
    \caption{\textbf{Process Matrix via conditioning}. Any Process matrix $W$ on Alice ($A$) and Bob ($B$) can be realised using a circuit consisting of: two initial maximally entangled states $\Phi^+_{\Ads_I}$ and $\Phi^+_{\Bds_I}$ (where we have used the shorthand $\Xds_I = X_IX'_I$), an ancillary state $\ket{0}_{C_I'}$, a unitary map (with corresponding Choi matrix $U$) that acts on all of the involved spaces, and a final conditioning on a measurement of the additional degrees of freedom $C_I$. This set-up, together with the comb $\Upsilon_{ABC_I}$ (with a final output line on $C_I$) one would receive without conditioning (outlined in magenta) is displayed in the bottom of the figure. Conditioning on measurement outcomes (here, $0$ and $1$) on $C_I$ then yields the process matrices $W^{(0)}$ and $W^{(1)}$ (middle of the figure). Choosing $U$ accordingly for the desired $W$, the probabilities obtained by conditioning on, say, outcome $0$, then coincide with those that one would obtain from $W$ (see\eq~\eqref{eqn::Equality_Prob}). Graphically, the resulting temporal Born rule is depicted as a process matrix with two `slots', with the respective CP maps inserted into those slots. Discarding $C_I$, i.e., combining the conditioned process matrices $W^{(0)}$ and $W^{(1)}$ yields the (causally separable) process matrix $W^{A\|B} = \rho_{A_IB_I} \otimes \ident_{A_OB_O}  = \ident_{AB}/{d_{A_I}d_{B_I}}$ (top of the figure). Note that time flows from bottom to top.}
    \label{fig::Conditioning_Circuit}
\end{figure}

While any process matrix $W$ can be realised by means of the above procedure, it is \textit{a priori} unclear, what properties the comb $\Upsilon_{ABC_I}$ has to satisfy in order to realise process matrices with different properties, such as causal non-separability. In this paper we study the properties of this comb and in particular the different choices in the conditioning itself. For example, as we shall see in Sec.~\ref{sec::CausReal}, the combs used to realise arbitrary process matrices do not have to display quantum correlations (in the splitting $AB:C_I$), even if the realised $W$ is causally non-separable. On the other hand, while entanglement is not necessary, entanglement in the splitting $AB:C_I$ still proves useful to increase the robustness for realising causal non-separability (see Sec.~\ref{sec::CausReal}). 

It is worth pointing out the similarities and differences of our procedure with the quantum switch~\cite{chiribella_quantum_2013, procopio_experimental_2015, goswami_indefinite_2018} and the protocols that use it~\footnote{Besides a system $C_I$ that is used as a control, the quantum switch can also have an additional non-trivial input and/or output wire. This difference to our scheme is not relevant for the reasoning we employ here.}. In order to see an enhancement in, for example, communication scenarios~\cite{ebler_enhanced_2018,chiribella_indefinite_2018, wilson_quantum_2020, procopio_communication_2019, guo_experimental_2020}, it is -- just like in our procedure -- crucially important to be in possession of the control qubit~\cite{ebler_enhanced_2018, costa_no-go_2020} (in our case, the system $C_I$). However, there are two main differences: on the one hand, our scheme requires conditioning. On the other hand, while in our case the remaining process is of a \textit{definite} causal order when the control qubit is discarded, in the case of the switch, the remaining process is a convex mixture of opposing causal orders. We will return to this latter distinction between the quantum switch and our conditioning procedure in Sec.~\ref{sec::CondCausOrder}, where we discuss the relation of causal order and conditioning and demonstrate that conditioning may lead to different \textit{definite} causal orders, making causality itself basis-dependent.

\subsection{Causal robustness}
In order for us to carry out our investigation, and to be able to quantify how far a given process matrix deviates from the set of causally separable ones, it is necessary to introduce a measure that allows us to gauge the causal non-separability of a process. One possible way to do so is the causal robustness~\cite{araujo_witnessing_2015} $\Ccal_R(W)$ that measures how much worst-case noise can be mixed with a given process matrix $W$ before it becomes causally separable: 
\begin{gather}
\label{eqn::Robustness}
\begin{split}
    &\Ccal_R(W) \\
    &=\!  \min \{s\geq 0 | \tfrac{W+ sW'}{1+s} \!=\! q W^{A\prec B} \!+\! (1-q) W^{B\prec A}\},
    \end{split}
\end{gather}
for some proper process matrix $W'$, some causally ordered process matrices $\{W^{A\prec B}, W^{B\prec A}\}$, and some probability $q$. Evidently, $\Ccal_R(W)$ vanishes iff $W$ is causally separable. Besides satisfying some reasonable desiderata one would require from a measure of causal non-separability (such as monotonicity under local unitary operations~\cite{araujo_witnessing_2015}), $\Ccal_R$ is amenable to efficient numerical evaluation, as it can be phrased as a semidefinite program \textbf{(SDP)}~\cite{araujo_witnessing_2015}. We provide this SDP, which we will use throughout to quantify the causal non-separability of the process matrices we consider, in App.~\ref{app:CausRob}.

\section{Basis-independent causal non-separability}
\label{sec::CausalCoherence}
In Ref.~\cite{milz_entanglement_2018}, an explicit example was given for conditionally realising the causally non-separable four-qubit process matrix 
\begin{gather}
    W^{\mathrm{(OCB})} = \tfrac{1}{4}[\ident_{AB} + \tfrac{1}{\sqrt{2}}(\sigma^z_{A_O} \sigma^z_{B_I} + \sigma^z_{A_I} \sigma^x_{B_I} \sigma^z_{B_O})]\, ,
\end{gather}
where $\{\sigma^x,\sigma^z\}$ are Pauli matrices on the respective spaces, and we have omitted the tensor products and identity matrices. We will denote this particular process matrix $W^{(\mathrm{OCB})}$ after the authors of Ref.~\cite{OreshkovETAL2012}, where it was first introduced. Using the SDP provided in the Appendix, the causal robustness of $W^{(\mathrm{OCB})}$ can be computed to be $\Ccal_R(W^{(\mathrm{OCB})}) \approx 0.17$. 

The parallel circuit, which allows one to realise $W^{(\mathrm{OCB})}$ with probabiltiy $p_C(0) = 1/2$ yields the causally ordered comb
\begin{gather}
\label{eqn::WocbIncoherent}
    \Upsilon_{ABC_I} \!=\! \tfrac{1}{2}W^{(\mathrm{OCB})} \otimes \ketbra{0}{0}_{C_I} \!+\! \tfrac{1}{2}W^{\#} \otimes \ketbra{1}{1}_{C_I}\, ,
\end{gather}
where $W^{\#} =  \tfrac{1}{2}\ident_{AB} - W^{(\mathrm{OCB})}$ is also causally non-separable~\cite{milz_entanglement_2018}; conditioning on the outcome $0$ when measuring the system $C_I$ in the computational basis then yields the process matrix $W^{(\mathrm{OCB})}$. Interestingly, Ref.~\cite{milz_entanglement_2018} proves that, in order to realise a causally non-separable process matrix, the total initial state in Fig.~\ref{fig::Conditioning_Circuit} must be genuinely entangled across all three parties $ABC$, and the unitary $U$ must also have entangling power. On the other hand, the resultant comb $\Upsilon_{ABC_I}$ of\eq~\eqref{eqn::WocbIncoherent} displays \textit{no} quantum correlations in the splitting $AB:C_I$.

While the set of combs with only classical correlations in the pertinent splitting allows for the realisation of causally non-separable process matrices, Eq.~\eqref{eqn::WocbIncoherent} raises the question \textit{what happens if there is entanglement between the conditioning qubit and the remaining degrees of freedom?} Put differently, a generic comb $\Upsilon_{ABC_I}$ will contain genuine quantum correlations across the partitions i) $A$ and $C_I$, ii) $AB$ and $C_I$ -- both corresponding to genuine quantum memory~\cite{giarmatzi_witnessing_2018} -- and iii) $B$ and $C_I$, which corresponds to a direct quantum cause (i.e., a quantum channel) between Bob and Charlie. These correlations constitute a useful resource for, amongst others, realising causally indefinite processes. 

It is easy to see that the causal non-separability of the resulting $W$ critically depends on the measurement basis. For example, in the above scenario, conditioning with respect to a measurement in the $\{\ket{\pm}\}$ basis yields the two process matrices $W^{(+)} = W^{(-)} = \tfrac{1}{4} \ident_{AB}$, which are causal. Put differently, `looking' at the process in different bases yields different resulting (conditional) causal structures and makes the property of causal non-separability vanish. Adding entanglement between the control qubit and the remaining degrees of freedom might help making this conditioning procedure more stable (in a sense defined below), potentially leading to scenarios where, independent of the respective measurement basis, conditioning \textit{always} leads to causally non-separable resulting processes. We explore this question in detail in Sec.~\ref{sec::Entanglement}, and further explore the basis dependence of causal ordering in Sec.~\ref{sec::CondCausOrder}. 

On the other hand, $W^{(\mathrm{OCB})}$ by means of a comb that does not display quantum correlations in the splitting $AB:C_I$, which raises the complementary questions, to the one above, \textit{could all process matrices can be obtained without quantum correlations as in the above splitting? If so, how do we interpret causally non-separable process matrices?} We start with this latter questions.

\subsection{Heralded Causal non-separability}
\label{sec::CausReal}

In general, the absence of entanglement between the conditioning system and the relevant degrees of freedom implies that measurements on $C_I$ merely herald pre-existing objects, but do not `create' them. In particular, for a comb of the form\eq~\eqref{eqn::WocbIncoherent}, a computational basis measurement on $C_I$ is noninvasive, suggesting that observing any one of the two possible observed outcomes reveals which of the two causally non-separable processes was `realised' in an individual run. Interestingly, such a realisation scheme without entanglement between $AB$ and $C$ exists for \textit{any} process matrix $W$ and we have the following Observation:
\begin{observation}
\label{obs::ElementReal}
For any process matrix $W\in \Bcal(\Hcal_A\otimes \Hcal_B)$ there exists a probability $p>0$ and a proper process matrix $W^\prime$, such that 
\begin{gather}
\label{eqn::ObsElement}
    \Upsilon_{ABC_I} = pW \otimes \ketbra{0}{0}_{C_I} +  (1-p) W^\prime \otimes \ketbra{1}{1}_{C_I} 
\end{gather}
is a causally ordered comb with $C_I$ as the last party. 
\end{observation}
\begin{proof}
For the proof, we first note that $\ident_{AB}/d_{A_IB_I}$ is a proper process matrix (with causal ordering $A\|B$). Given any process matrix $W \in \Bcal(\Hcal_A\otimes \Hcal_B)$ there always exists a probability $p>0$ such that $(\ident_{AB}/d_{A_IB_I} - pW) = (1-p)W' \geq 0$, where the factor $(1-p)$ is introduced for correct normalization of $W'$. It is easy to see that if $W$ is a proper process matrix, then so is $W'$. Setting $\Upsilon_{ABC_I} = pW \otimes \ketbra{0}{0}_{C_I} +  (1-p) W^\prime \otimes \ketbra{1}{1}_{C_I}$, we see that -- by construction -- $\tr_{C_I}\Upsilon_{ABC_I} = \ident_{AB}/d_{A_IB_I}$ and $\Upsilon_{ABC_I} \geq 0$; thus, $\Upsilon_{ABC_I}$ satisfies the causality constraints of Eq.~\eqref{eqn::Causality3}, implying that it is a causally ordered comb.
\end{proof}

Since all process matrices can be implemented by means of a comb of the form of \eq~\eqref{eqn::ObsElement}, an interesting interpretation suggests itself: causally non-separable processes can be `present' in individual runs of an experiment (and are heralded by the measurement outcome on $C$); but this possible absence of causal order is washed out on average, i.e., when the system $C$ is discarded.

At first glance, this latter statement might not seem surprising. Indeed, let us consider the spatial case; there, when measuring, e.g., the completely mixed state $\ident/2$ of a qubit in the computational basis, one will obtain outcomes $0$ (with corresponding state $\ket{0}$) and $1$ (with corresponding state $\ket{1}$) with equal probability, suggesting that half of the time (whenever outcome $0$ was obtained) the state $\ket{0}$ was prepared and the other half of the time (whenever outcome $1$ is obtained) the state $\ket{1}$. However, the state $\ident/2$ could, for example, be understood in the same vein as an equal mixture of $\ket{+}$ and $\ket{-}$, invalidating this interpretation, and seemingly casting doubt on the interpretation we provided above for processes. However, in the process case we discussed above, the conditioning happens on an \textit{external} system. The better analogy in the spatial case would thus be to consider a machine that prepares states $\ket{0}$ and $\ket{1}$ with equal probabilities, but, whenever preparing either of these states, it attaches a corresponding flag to it, resulting in the overall state 
\begin{gather}
    \rho = \tfrac{1}{2} \ketbra{0}{0} \otimes \ketbra{0}{0}_{\rm flag} + \tfrac{1}{2} \ketbra{1}{1} \otimes \ketbra{1}{1}_{\rm flag}
\end{gather}
Discarding the flag would, again, yield the maximally mixed state. However, now, by measuring the flag, the obtained outcome heralds the state that was prepared in the respective run of the experiment. In the same vein, the flag on $C_I$ in Eq.~\eqref{eqn::ObsElement} can be considered as revealing the process matrix of the respective run of the experiment.

Evidently, as the notion of a causal order (or the absence thereof) is not well-defined for an \textit{individual} run of an experiment, but rather corresponds to a statistical statement over many runs, such an interpretation has to be taken with care. This fact notwithstanding, for the above comb $\Upsilon_{ABC_I}$ of\eq~\eqref{eqn::ObsElement}, obtaining one of the outcomes $\{0,1\}$ when measuring $C_I$ in the computational basis can be interpreted as revealing which of the process matrices $\{W, W'\}$ was employed in the respective run of the experiment. In other words, in each run of the experiment there is no causal order between $A$ and $B$, whether $C_I$ is observed or not. For a given outcome on $C_I$, we cannot even attribute probabilistic causal direction between $A$ and $B$. We only see a causally ordered process on average due to our ignorance of the measurement on $C_I$. Importantly, this interpretation would not hold if there was entanglement, or other quantum correlations, in the splitting $AB:C_I$; in this case, measurements on $C_I$ would not simply reveal a pre-existing property as the system $C_I$ would not only be a heralding flag.

While the same arguments could seemingly be made for conditioning on \textit{any} positive matrix $W$ -- valid process or not -- there is, as we already emphasized above, a fundamental difference between process matrices that are valid, and those that are not; for the former, the conditioning probability is \textit{independent} of the instruments that Alice and Bob employ, while for the latter the probability to obtain outcome $0$ or $1$ will \textit{always} depend on Alice's and Bob's instruments~\cite{silva_connecting_2017, milz_entanglement_2018}, making the conditioning procedure somewhat ill-defined in the sense mentioned above; additionally, such a dependence on the instrument would prevent an interpretation of the `realised' positive matrix $W$ as a pre-existing object, but rather one that appears to be contingent on Alice's and Bob's respective instruments.

On the other hand, \textit{any} conditioning procedure of the quantum comb of\eq~\eqref{eqn::ObsElement} on $C_I$ will yield a proper process matrix, making such conditioning scenarios well-defined. However, as previously mentioned, in general, not all such conditioning will lead to causally non-separable process matrices, even if conditioning in the computational basis does. Next, we will show that the range of bases that lead to causally non-separable process matrices can be extended when entanglement is present, making correlations between $AB$ and $C_I$ a robustness resource. 

\subsection{Entanglement and causal non-separability}
\label{sec::Entanglement}
In the previous sections, we discussed (the comb of) a concrete circuit for the realisation of $W^{(\mathrm{OCB})}$ by means of measurements on $C_I$. Here, starting from this concrete circuit and the corresponding $\Upsilon_{ABC_I}$, we investigate how `robust' such a procedure can be made by adding entanglement between $C_I$ and $AB$.

Naturally, the causality properties of the conditioned process matrices depend on the employed conditioning basis. For example, conditioning the comb $\Upsilon_{ABC_I}$ of Eq.~\eqref{eqn::WocbIncoherent} in the $\{\ket{\pm}_{C_I}\}$ basis yields a causally separable process, as 
\begin{gather}
    \tr_{C_I}(\Upsilon_{ABC_I} \ketbra{\pm}{\pm}_{C_I}) \propto \ident_{AB}\, .
\end{gather}
Consequently, here, by `robust' we mean
the range of conditioning bases for which the resulting process matrix is still causally non-separable. Using the ideas developed above, we show that it is possible to devise a circuit that yields a causally non-separable process matrix for \textit{any} conditioning basis.

\begin{figure*}[t]
\centering
\subfigure[$~\{c_{11},c_{15},c_{51}\} = \{0,0,0\}$]
{
\includegraphics[width=0.28\linewidth]{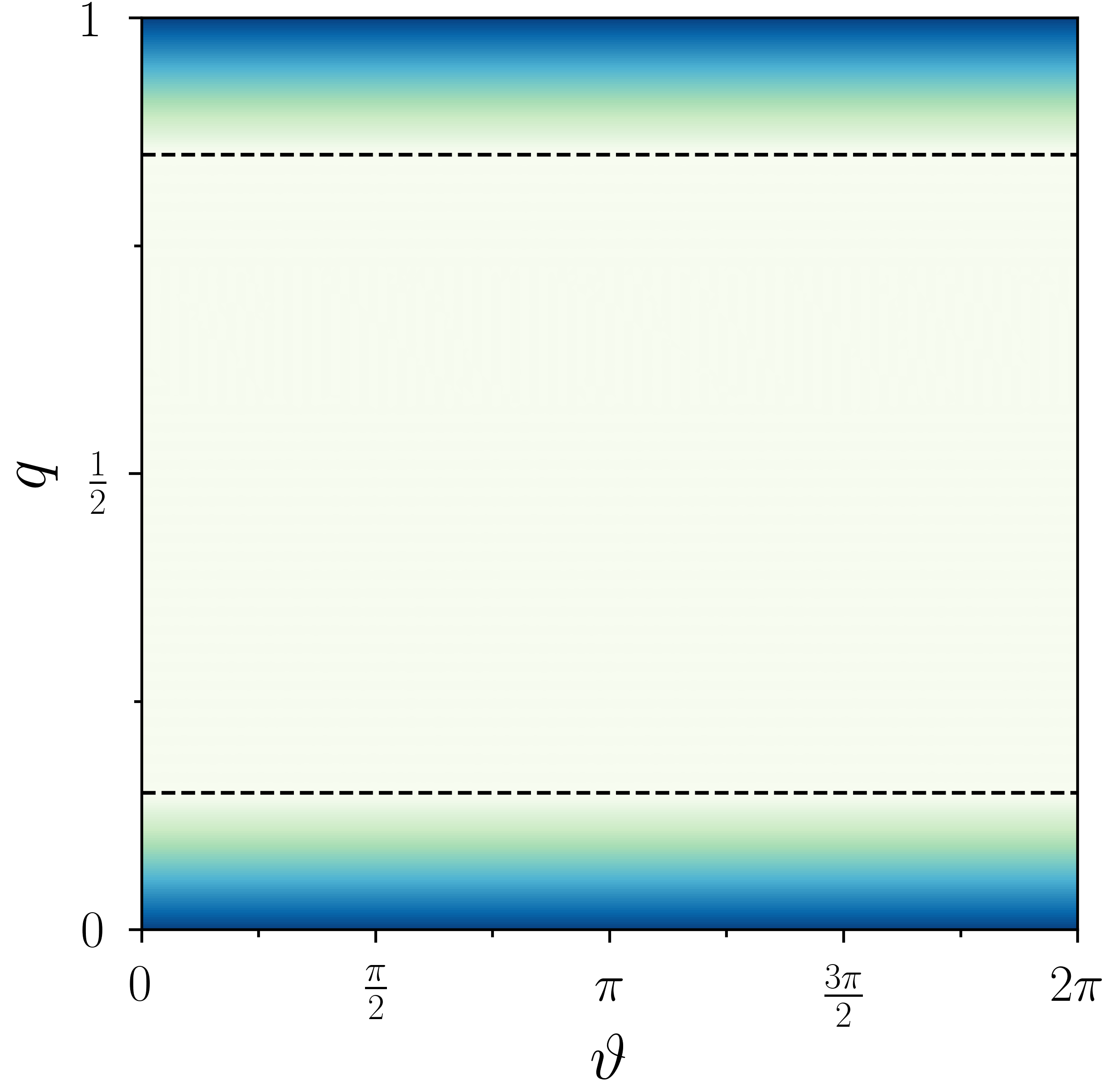}
\label{fig::HeatPlot000}
}\hfill
\subfigure[$~\{c_{11},c_{15},c_{51}\} = \{\tfrac{1}{4},0,0\}$]
{
\includegraphics[width=0.28\linewidth]{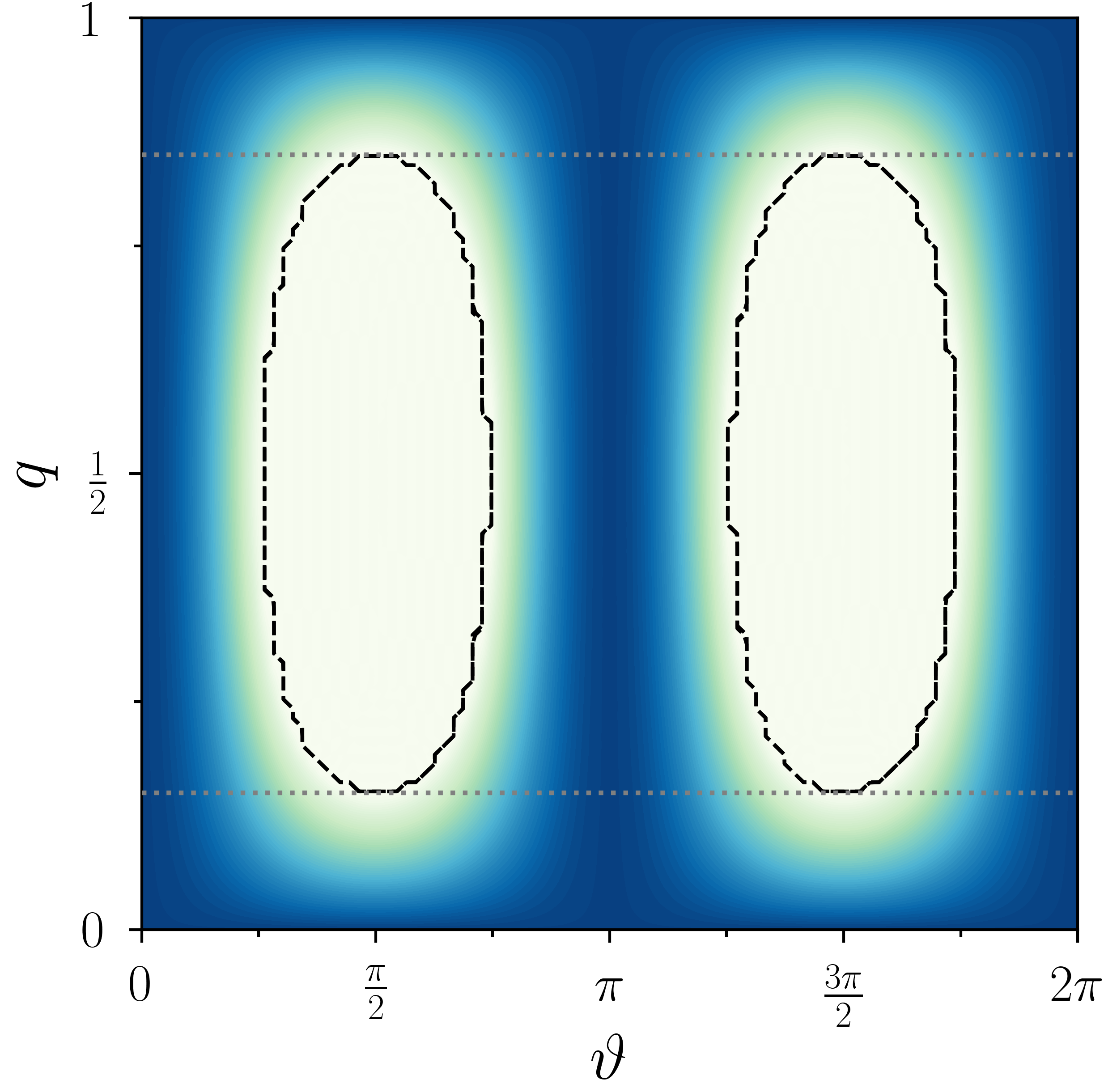}
\label{fig::HeatPlot1/400}
}\hfill
\subfigure[$~\{c_{11},c_{15},c_{51}\} =\{\tfrac{1}{4\sqrt{2}},\tfrac{1}{4\sqrt{2}},\tfrac{1}{4\sqrt{2}}\}$] 
{ 
\includegraphics[width=0.28\linewidth]{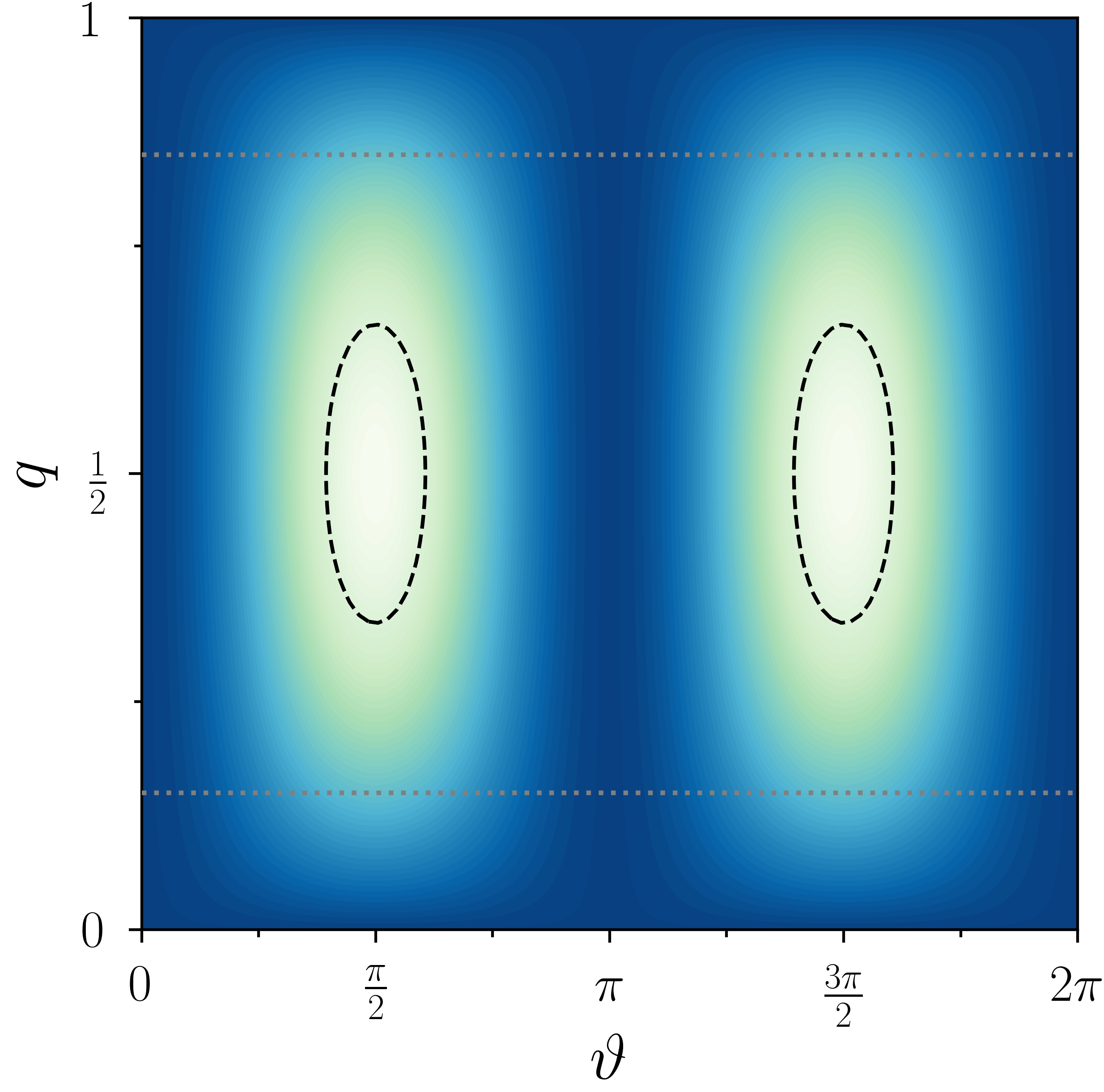}
\label{fig::HeatPlot1/41/41/4}
}
\subfigure
{ 
\includegraphics[width=0.0485\linewidth]{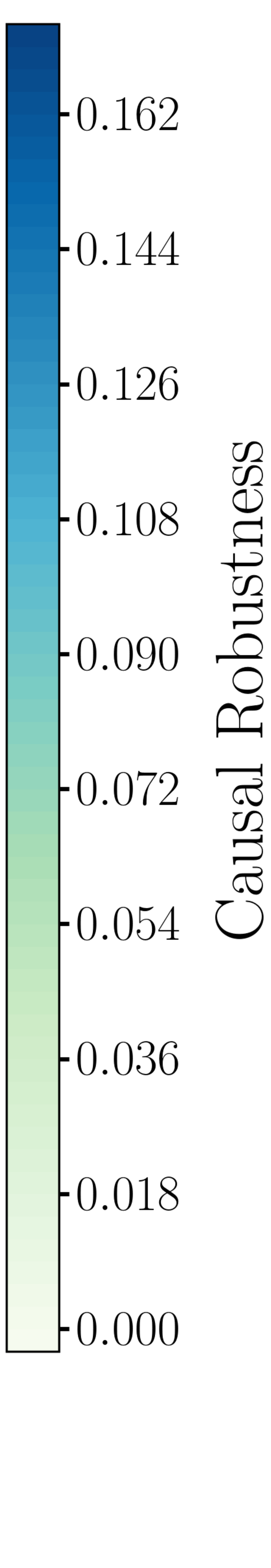}
}
\caption{\textbf{(a) Causal Robustness for $\{c_{11},c_{15},c_{51}\} = \{0,0,0\}$}. Conditioning on $\ket{0}$ ($\ket{1}$) corresponds to the lines $q=0$ ($q=1$), where the causal robustness is maximal (for the chosen conditioning set-up). Biasing the conditioning basis towards a superposition of $\ket{0}$ and $\ket{1}$, i.e., increasing $q$ from $q=0$ or decreasing it from $q=1$ then quickly leads to causally separable process matrices. The causal robustness of the conditioned process matrices is independent of the angle $\vartheta$. For reference, the lines $q=0.15$ and $q=0.85$  (dotted lines) where the causal robustness drops to zero, respectively, are added. \textbf{(b) and (c) Causal Robustness for $\{c_{11},c_{15},c_{51}\} = \{\tfrac{1}{4},0,0\}$} and $ \{c_{11},c_{15},c_{51}\} =\{\tfrac{1}{4\sqrt{2}},\tfrac{1}{4\sqrt{2}},\tfrac{1}{4\sqrt{2}}\}$ (evaluated on a $100\times 100$ grid). In both cases, the parameter range for which $W(q,\vartheta)$ is causally non-separable is significantly increased with respect to Fig.~\ref{fig::HeatPlot000}, and the causal robustness of the conditioned process matrices depends on the angle $\vartheta$. For reference and comparison, the contours $\Ccal_R(W(q,\vartheta))=0$ (black dotted lines) as well as the corresponding lines from panel (a) (gray dotted lines) are shown, respectively.}
\end{figure*}

To start with, we consider the causally ordered comb
\begin{gather}
\begin{split}
    \Upsilon^F_{ABC_I} &= \tfrac{1}{2} W^{(\mathrm{OCB})}\otimes \ketbra{0}{0}_{C_I} + \tfrac{1}{2} W^{\#}\otimes \ketbra{1}{1}_{C_I} \\
    \label{eqn::FTerms}
    &\phantom{=}+ F\otimes \ketbra{0}{1}_{C_I} + F^\dagger \otimes \ketbra{1}{0}_{C_I} \, ,
    \end{split}
\end{gather}
where $W^{\#} =  \tfrac{1}{2}\ident_{AB} - W^{(\mathrm{OCB})}$. If $F = 0$, $\Upsilon^F_{ABC_I}$ is separable in the splitting $C_I:AB$, and we recover the original parallel circuit scenario of\eq~\eqref{eqn::WocbIncoherent} for the realisation of $W^{(\mathrm{OCB})}$. If $\Bcal(\Hcal_{A} \otimes \Hcal_B) \ni F \neq 0$, then -- as long as $F$ leads to a valid process (see below) -- $\Upsilon^F_{ABC_I}$ is (generally) entangled and potentially more robust, in the above sense, against a change of conditioning basis. To see this more clearly, consider the process matrix $W(q,\vartheta)$ obtained from conditioning $\Upsilon^F_{ABC_I}$ on a measurement outcome corresponding to the general pure qubit state $\ket{\Phi(q,\vartheta)} = \sqrt{q}\ket{0}_{C_I} + \sqrt{1-q}e^{\iu \vartheta}\ket{1}_{C_I}$. As $F$ in\eq~\eqref{eqn::FTerms} has to be traceless for $\Upsilon^F_{ABC_I}$ to be positive (see App.~\ref{app::Fterms}), the conditioning probability is equal to $1/2$ and we have
\begin{gather}
    \begin{split}
  W(q,\vartheta) &= qW^{(\mathrm{OCB})} + (1-q)W^\# \\
  &\phantom{=}+ 2\sqrt{q(1-q)} (e^{\iu \vartheta} F + e^{-\iu \vartheta} F^\dagger)\, .
\end{split}
\end{gather}
Choosing a non-vanishing $F$ can now drastically increase the range of parameters $(q,\vartheta)$ for which $W(q,\vartheta)$ is causally non-separable, as compared to the case $F=0$.

Before continuing, it is worth discussing why a circuit that realises $W^{(\mathrm{OCB})}$ is a good starting point for the analysis we aim to conduct. While $W^{(\mathrm{OCB})}$ is not the process matrix that maximizes the causal robustness for the case of two parties and qubit systems~\cite{Bava}, it has some appealing properties that make it a good candidate for such an investigation. On one hand, while not maximal, its causal robustness is nevertheless large. More importantly still, it is of rank $8$ (which is half of the full rank) and all of its eigenvalues are equal to $\tfrac{1}{2}$, such that
\begin{gather}
    W^{(\mathrm{OCB})} W^\# = W^{(\mathrm{OCB})} (\tfrac{1}{2} \ident - W^{(\mathrm{OCB})}) = 0\, ,
\end{gather}
which significantly simplifies the following considerations. In particular, using the eigendecomposititons $W^{(\mathrm{OCB})} = \tfrac{1}{2} \sum_{i=1}^8 \ketbra{\Psi_i}{\Psi_i}$ and $W^\# = \tfrac{1}{2} \sum_{j=1}^8 \ketbra{\Psi^\perp_j}{\Psi_j^\perp}$, where $\braket{\Psi_i|\Psi_{i'}} = \delta_{ii'}$, $\braket{\Psi_j^\perp|\Psi_{j'}^\perp} = \delta_{jj'}$, and  $\braket{\Psi_i|\Psi_{j}^\perp} = 0$, we show in App.~\ref{app::Positivity} that $F$ needs to  be of the form 
\begin{gather}
\label{eqn::Parameters}
F = \sum_{i,j=1}^8 (c_{ij} \ketbra{\Psi_i}{\Psi_j^\perp} + d_{ij} \ketbra{\Psi_i}{\Psi_j^\perp})\, ,
\end{gather}
with $c_{ij}, d_{ij} \in \mathbbm{C}$, for $\Upsilon^F_{ABC_I}$ to be positive (naturally, not \textit{all} $c_{ij}, d_{ij} \in \mathbbm{C}$ lead to positive $\Upsilon^F_{ABC_I}$). Additionally, in order for $\Upsilon^F_{ABC_I}$ to be positive, it is necessary that all coefficients $d_{ij}$ vanish (see App.~\ref{app::Positivity}). Finally, imposing that conditioning on \textit{any} $\ket{\Phi(q,\vartheta)}$ yields a \textit{proper} process matrix, i.e., one that satisfies\eq~\eqref{eqn::Def_W} allows us to further reduce the number of free parameters in\eq~\eqref{eqn::Parameters}. In App.~\ref{app::validF}, we show that there are three free parameters $\{c_{11},c_{15},c_{51}\}$ that remain, while all other parameters $c_{ij}$ either vanish or are determined by the choice of those three parameters. Consequently, choosing a triple $\{c_{11},c_{15},c_{51}\}$, computing the remaining parameters according to the conditions worked out in App.~\ref{app::validF}, and checking that the resulting $\Upsilon^F_{ABC_I}$ is positive then ensures that every conditioned $W(q,\vartheta)$ that results from it is a proper process matrix. Having reduced the number of free parameters down to three thus provides a good test-bed to investigate the stability of the conditioning procedure against changes in the conditioning basis. 

Below we explore this parameter space in some detail for the interested reader (others may wish to directly move to Obs.~\ref{obs:obsindcauinsp}, which is our second main result).
To this end, in order to establish a baseline, we first provide the conditioning results for the case $\{c_{11},c_{15},c_{51}\} = \{0,0,0\}$, i.e., $F = 0$. As already mentioned, in this case, the resulting conditioned process matrix is definitely causally separable for $\Phi(q = \tfrac{1}{2},\vartheta = 0) = \ket{+}$. However, as can be readily seen from the corresponding plot, in Fig.~\ref{fig::HeatPlot000}, of the causal robustness with respect to the conditioning parameters $q$ and $\vartheta$, the conditioned process matrices are causally separable for a large range of parameters, and only become causally non-separable when $\Phi(q,\vartheta)$ is sufficiently close to $\ket{0}$ or $\ket{1}$. More concretely, the causal robustness decreases with $|q-\tfrac{1}{2}|$, and $W(q,\vartheta))$ becomes causally separable at $q\approx 0.85$ and $q\approx0.15$, respectively. Additionally, due to the absence of off-diagonal terms when $F=0$, the angle $\vartheta$ of the state $\ket{\Phi(q,\vartheta)}$ has no influence on the causal robustness of the resulting process matrices $W(q,\vartheta)$.

Having established this baseline, we can now analyse the influence of non-vanishing terms $F$, and thus -- at least in all the cases we consider -- non-vanishing entanglement between $C_I$ and $AB$. First, for simplicity, we set $c_{15} = c_{51} = 0$. In this case, as we show in App.~\ref{app::Pos_Rev}, we must have $|c_{11}| \leq \tfrac{1}{4}$ for $\Upsilon^F_{ABC_I}$ to be positive. A natural choice is thus $\{c_{11},c_{15},c_{51}\} = \{\tfrac{1}{4},0,0\}$. The causal robustness of the resulting process matrices $W(q,\vartheta)$ is displayed in Fig.~\ref{fig::HeatPlot1/400}. With respect to the results for $\{c_{11},c_{15},c_{51}\} = \{0,0,0\}$, the parameter space for which $W(q,\vartheta)$ is causally non-separable is significantly increased. While, as before, $W(q,\vartheta)$ is still causally non-separable for $q \in [0.85,1]$ and $q \in [0.15,1]$, now, depending on the angle $\vartheta$, there are causally non-separable process matrices for all values of the parameter $q$.

We can achieve even better results, i.e., a wider range of parameters, for which $W(q,\vartheta)$ is causally non-separable, by choosing all of the coefficients $\{c_{11},c_{15},c_{51}\}$ to be the same (and equal to $c$). As we show in App.~\ref{app::Pos_Rev}, this implies $|c| \leq \tfrac{1}{4\sqrt{2}}$. The corresponding results for the choice $\{c_{11}, c_{15}, c_{51}\} = \{\tfrac{1}{4\sqrt{2}}, \tfrac{1}{4\sqrt{2}}, \tfrac{1}{4\sqrt{2}}\}$ are shown in Fig.~\ref{fig::HeatPlot1/41/41/4}. 

Given that the two previous choices for the coefficients $\{c_{11}, c_{15}, c_{51}\}$ yield process with low causal robustness on the line $q=\tfrac{1}{2}$, it appears natural to search for coefficients that `maximize' the causal robustness along said line, i.e., the coefficients, for which
\begin{gather}
\displaystyle \min_{\vartheta}[\Ccal_R(W(\tfrac{1}{2},\vartheta))]    
\end{gather}
is maximized (and non-vanishing). Given that such an  optimization requires the solution of a large number of SDPs for each choice of $\{c_{11},c_{15},c_{51}\}$, it is out of reach for the full parameter space of allowed coefficient triplets. However, focusing on the family $\{c_{11},e^{\iu \varphi_1}c_{11},e^{\iu \varphi_2}c_{11}\}$, with $c_{11} \in \mathbbm{R}$, allows one to find a choice of coefficients that likely leads to conditioned process matrices $W(q,\vartheta)$ that are causally non-separable on the line $q=\tfrac{1}{2}$, and, potentially, also on the remaining space of conditioning parameters $\{q,\vartheta\}$. We provide the conditions on $|c_{11}|$ for said family to yield a positive $\Upsilon^F_{ABC_I}$ in App.~\ref{app::Pos_Rev}

Following this approach, we find that a good candidate for coefficients that are optimal in the above sense is given by $\{c_{11},c_{15},c_{51}\} = \{\tfrac{1}{8},-\tfrac{1}{8},\tfrac{1}{8}\}$ (see Fig.~\ref{fig::Heat_Plot_ideal} for the corresponding heat plot). 
\begin{figure}[t]
    \centering
    \includegraphics[width=0.85\linewidth]{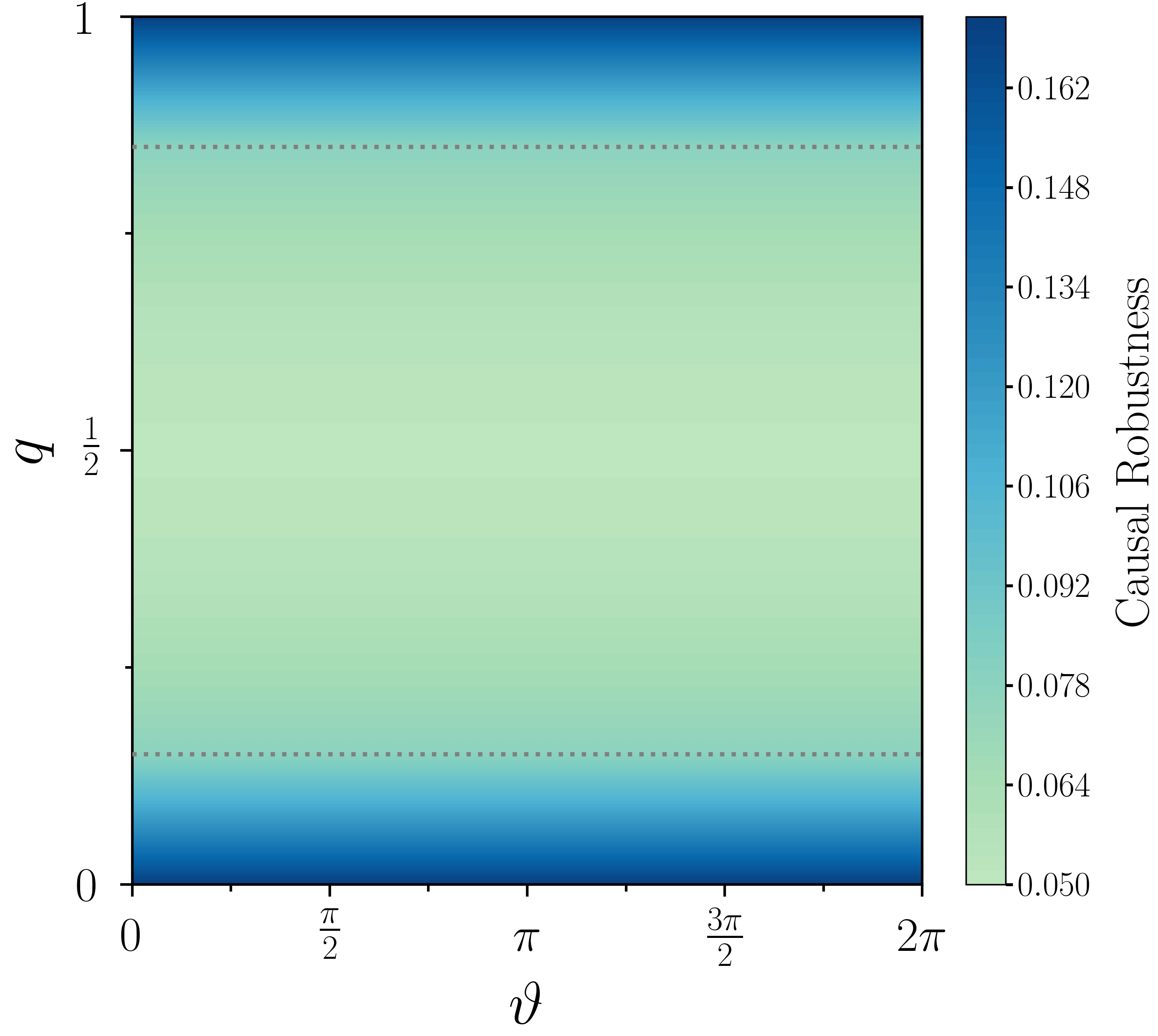}
    \caption{\textbf{Causal Robustness for $\{c_{11},c_{15},c_{51}\} = \{\tfrac{1}{8},\tfrac{-1}{8},\tfrac{1}{8}\}$}. The conditioned process matrices have non-vanishing causal robustness ($\geq 0.051782$) for all values of $q$ and $\vartheta$ (evaluated on a $100\times 100$ grid), i.e., for all possible conditionings.}
    \label{fig::Heat_Plot_ideal}
\end{figure}
For this choice of coefficients, as is obvious from Fig.~\ref{fig::Heat_Plot_ideal}, \textit{all} conditioned process matrices are causally non-separable. We provide a proof of this statement in App.~\ref{app::Allnonsep}. This leads to the following observation: 
\begin{observation}\label{obs:obsindcauinsp}
There are causally ordered combs $\Upsilon^F_{ABC_I}$ that lead to causally non-separable conditioned process matrices $W(q,\vartheta)$ for conditioning in any basis. 
\end{observation}
See Fig.~\ref{fig::Comic} for a graphical representation. While the above observation \textit{a priori} only holds true for conditioning with rank-one measurements, we can even allow for some noise in the conditioning process. Numerically, the causal robustness of $W(q,\vartheta)$ never falls below $0.051782$, implying that there is a $5\%$ robustness against worst case noise in the measurement procedure, before the least robust $W(q,\vartheta)$ becomes causally non-separable. 

Crucially, the above result implies that it is not necessary for the conditioning party $C$ to be perfectly aligned with the remaining two parties $AB$ in order for causal non-separability to occur. To see this, consider a situation where only conditioning a single fixed basis yields causally non-separable process matrices. This case would require an observer $C$ that is perfectly aligned (or knows in what sense they are misaligned) with said basis in order for causal non-separability to be observed. 

Here, on the other hand, \textit{any} observer that conditions the system in some (arbitrarily chosen) basis would create a causally non-separable process on Alice and Bob. Consequently, as no particular alignment of the conditioning party is required, we will call this effect \textit{basis-independent}. Importantly though, while it is \textit{basis-independent}, i.e., independent of the pure measurement that is carried out, the causal non-separability of the conditioned process matrices is \textit{not} device-independent. Since tracing out the degrees of freedom $C_I$ yields a causally ordered process, there always exists a trivial POVM 
\begin{gather}
\{E_0 = \tfrac{1}{2}\ident_{C_I},E_1 = \tfrac{1}{2}\ident_{C_I}\}
\end{gather}
such that both `outcomes' yield a causally ordered process.

Our above results establish causal non-separability as a property that can exist in a basis-independent manner. In the next section, we will again make use of the coherence terms $F$ to realise processes where the causal order is basis-dependent. With the invariance of causal order under change of basis in mind, we now turn our attention to the inverse question: can causal order itself be basis-dependent? While special relativity forbids such an effect, we will see that within the conditioning framework we use, such a basis-dependence is indeed possible.

\section{Basis-dependent causal order}
\label{sec::CondCausOrder}

Up to this point, we have considered conditioning scenarios that were designed so that they yield causally non-separable processes, and we were interested in the stability with respect to the choice of conditioning basis. Here, we abandon these considerations of robustness and ask the related question: Can causal order itself be basis-dependent, i.e., are there processes where conditioning in one basis yields a process that is ordered $A\prec B$, while conditioning in a different basis yields a process that is ordered $B\prec A$? Here, we show that this is possible, both probabilistically, i.e., the respective conditioned processes only display the desired causal order when the `correct' outcome in Charlie's laboratory occurs, and, importantly, deterministically, i.e., the causal ordering of the observed conditioned processes only depends on the choice of measurement basis, but not on the respective outcomes. While the former scenario potentially allows for the realisation of a wider range of processes with opposing causal order, it is perfectly conceivable classically. However, the deterministic case is of foundational importance, as it admits the interpretation of causality as a measurement-dependent property, since the causal direction can be chosen at will by Charlie. Due to this contextual nature, the latter scenario is genuinely quantum.

A complementary, albeit formally different question with respect to the dependence of causal order on experimental observations has been considered in Ref.~\cite{pienaar_time-reversible}, where time-reversible (quantum) causal models and the influence of the observer on the perception of causal order were studied; there, the perceived causal structure with respect to the employed operations (in our notation, the operations $M_A^{(i)}$ and $M_B^{(j)}$) was analysed. In our work, the respective operations in Alice's and Bob's laboratory are unrestricted, and the respective causal order is contingent on the conditioning basis in Charlie's laboratory. Additionally, such a potential measurement-dependence of causal order is reminiscent of the quantum switch, with the crucial difference that the conditioning combs $\Upsilon_{ABC_I}$ we consider are causally ordered, while the switch is causally non-separable~\cite{araujo_witnessing_2015}. This, in turn, allows one to probabilistically condition onto opposing causal orders by means of one measurement basis, a feat not possible when causally ordered combs are employed (see below).

In what follows, when we consider causally ordered processes, we will mean definite causal order, i.e., not of the form $A\|B$, unless explicitly stated otherwise. Naturally, changing the conditioning basis changes the properties of the respective conditioned processes. In principle then, conditioning in two different bases might yield processes of opposing different orders. Importantly though, such an effect is indeed basis dependent and can only occur for two \textit{different} choices of conditioning bases; as we show below, it cannot be present when conditioning in only \textit{one} fixed basis with two different possible outcomes is considered.

\subsection{Opposite causal order for different conditioning bases}

We first show that, using two different conditioning bases, it is indeed possible to obtain processes of opposing causal orders. To this end, we make the following observation: 
\begin{observation}
\label{obs::OppCausOrd}
If two processes $\ident_{A_O} \otimes W_{A_IB}^{B\prec A}$ and $\ident_{B_O} \otimes W_{AB_I}^{A\prec B}$ of opposite definite causal order satisfy 
\begin{gather}
\begin{split}
    &p\ident_{A_O} \otimes W_{A_IB}^{B\prec A} \leq \ident_{B_O} \otimes W_{AB_I}^{A\prec B} \\
    \label{eqn::opp_caus_ord}
    \text{or} \quad &p\ident_{B_O} \otimes W_{AB_I}^{A\prec B} \leq \ident_{A_O} \otimes W_{A_IB}^{B\prec A}
\end{split}
\end{gather}
for some $0<p\leq 1$, then there exists a causally ordered process $\Upsilon_{ABC_I}$ such that conditioning on one of the outcomes when measuring in the $\{\ket{0},\ket{1}\}$ and $\{\ket{+},\ket{-}\}$ bases yields respective processes of opposing causal order. 
\end{observation}
\begin{proof}
We show this observation by explicit construction, focusing on the case $p\ident_{A_O} \otimes W_{A_IB}^{B\prec A} \leq \ident_{B_O} \otimes W_{AB_I}^{A\prec B}$. The other case follows in the same vein. We set 
\begin{gather}
\label{eqn::PropW_constr}
(1-p)W^\# := \ident_{B_O} \otimes W_{AB_I}^{A\prec B} - p\ident_{A_O} \otimes W_{A_IB}^{B\prec A}.
\end{gather} 
By assumption, $W^\#\geq 0$ holds, and it is easy to see that $W^\#$ is a proper process matrix. With this, we can define
\begin{gather}
\label{eqn::CausalityCombNoCorr}
    \begin{split}
    \Upsilon^{A\prec B \prec C_I}_{ABC_I} =& \ p\ident_{A_O} \otimes W_{A_IB}^{B\prec A} \otimes \ketbra{0}{0}_{C_I} \\ &+(1-p)W^\# \otimes \ketbra{1}{1}_{C_I}.
\end{split}
\end{gather}
Analogous to the proof of Obs.~\ref{obs::ElementReal}, we see that $\Upsilon^{A\prec B \prec C_I}_{ABC_I}$ satisfies the causality constraint of Eq.~\eqref{eqn::Causality3}, which implies that it is a causally ordered comb with ordering $A\prec B\prec C_I$. Conditioning on outcome $0$ (which occurs with probability $p$) when measuring in the computational basis yields the process matrix $\ident_{A_O} \otimes W_{A_IB}^{B\prec A}$ which is ordered $B\prec A$ by assumption. On the other hand, conditioning on outcome $+$ (corresponding to the projector $\ketbra{+}{+}_{C_I}$) when measuring in the $\{\ket{\pm}\}$ basis yields 
\begin{gather}
\begin{split}
    W^{(+)} &\propto \tfrac{1}{2}p\ident_{A_O} \otimes W_{A_IB}^{B\prec A} + \tfrac{1}{2}(1-p)W^\# \\
    &= \tfrac{1}{2}\ident_{B_O} \otimes W_{AB_I}^{A\prec B}\, ,
    \end{split}
\end{gather}
where we have used\eq~\eqref{eqn::PropW_constr}. As $W^{(+)}$ is thus ordered $A\prec B$, this concludes the proof.  
\end{proof}
It remains to show that there indeed exist two processes of opposing causal order, such that one of the Eqs.~\eqref{eqn::opp_caus_ord} is satisfied. Such processes are not hard to find. For example, if a process $\ident_{B_O} \otimes W_{AB_I}^{A\prec B}$ is of full rank, then for \textit{any} $B\prec A$ process $\ident_{A_0}\otimes W_{A_IB}^{B\prec A}$, by continuity, there exists a $p>0$ such that $\ident_{B_O} \otimes W_{AB_I}^{A\prec B} - p\ident_{A_0}\otimes W_{A_IB}^{B\prec A} \geq 0$. A simple Markovian~\cite{pollock_non-markovian_2018,pollock_operational_2018,1367-2630-18-6-063032,giarmatzi_witnessing_2018} example of a full rank $A\prec B$ process is
\begin{gather}
    W_{AB_I}^{A\prec B} =   \tfrac{1}{2}\ident_{A_I} \otimes [(r\widetilde\Phi^+_{A_OB_I} + \tfrac{(1-r)}{2}\ident_{A_OB_I})]\, ,
\end{gather}
where the unnormalized maximally entangled state $\widetilde\Phi^+_{A_OB_I}$ is the Choi matrix of the identity channel $\Ical_{A_O \rightarrow B_I}$. For $0<r<1$, the above process is of full rank and of causal order $A\prec B$, thus allowing for the realisation of two opposite causal orders for conditioning in two different bases (see Fig.~\ref{fig::Comic} for a graphical representation).

As before, somewhat surprisingly, the provided scenario does \textit{not} require any entanglement between $C_I$ and $AB$ in the employed causally ordered process $\Upsilon^{A\prec B \prec C_I}_{ABC_I}$. While it allows for the realisation of opposing causal orders by means of measurements in two different bases, this prescription has the obvious drawback that for the `unwanted' outcomes (here, $1$ and $-$), the realised process matrix does not possess the desired causal order. More specifically, in\eq~\eqref{eqn::PropW_constr}, $W^\#$ cannot be of causal ordering $A\prec B$, as otherwise\eq~\eqref{eqn::PropW_constr} could not hold (the sum of two process matrices of order $B\prec A$ cannot be of order $A\prec B$). Rather, $W^\#$ is either a mixture of causal orders or it is causally non-separable, implying that for the outcome $1$, the resulting process matrix is not of the desired order. This, then, renders the above scheme a probabilistic one with respect to a POVM.

Importantly, this caveat cannot be remedied in the absence of quantum correlations between $C_I$ and $AB$; if, for example, the process matrix $W^\#$ in\eq~\eqref{eqn::CausalityCombNoCorr} was of the same order as $W_{A_IB}^{B\prec A}$, the process matrix obtained for outcome $0$, then \textit{no} conditioning basis could lead to a process of opposite causal order; adding classical correlations would only lead to convex combinations of processes of order $B\prec A$, which, itself would again be a process of the same ordering. This situation changes drastically when correlations between $C_I$ and $AB$ are present in $\Upsilon_{ABC_I}^{A\prec B\prec C_I}$.

\subsection{Delayed-choice causal order}
As we have seen in Sec.~\ref{sec::Entanglement}, entanglement can vastly enhance the robustness for realising a causally non-separable processes. We show that it allows for causal order to be considered a basis-dependent quantity. While this was already an implication of Obs.~\ref{obs::OppCausOrd}, there, it was still a question of chance; not every outcome led to the desired causal order, implying that the causal order was not merely fixed by the choice of basis, but by the choice of basis \textit{and} the obtained measurement outcome. We now provide a scenario, where Charlie, by choosing the basis he measures in, can choose the direction of the causal order. In particular, we have the following observation
\begin{observation}
\label{obs::DetCausOrder}
Causal order can be instrument-dependent in a deterministic way, i.e., the causal order of the realised process matrices is fully determined by the respective choice of basis. 
\end{observation}
Before proving this observation, we emphasize the analogy to the results of Sec.~\ref{sec::Entanglement}. There, without added entanglement in the splitting $C_I:AB$, it appeared to not be possible to devise a scenario that led to causally non-separable process matrices for all conditioning basis. Here, entanglement allows us to overcome the limitations that apply for combs without the respective correlations and enables us to choose the causal order of the conditioned processes deterministically. We now prove the above Observation by providing an explicit example. 

\begin{proof}
To this end, consider a comb that yields two process matrices $W^{A\prec B}_{AB_I}\otimes \ident_{B_O}$ and $\widetilde{W}^{A\prec B}_{AB_I}\otimes \ident_{B_O}$ of order $A\prec B$ when conditioned in the $z$-basis, but has additional cross-terms $F\in\Bcal(\Hcal_A\otimes \Hcal_B)$:
\begin{gather}
\label{eqn::DetermCaus}
\begin{split}
    \Upsilon^{A\|B\prec C_I}_{ABC_I} 
    \!=& \tfrac{1}{2} (W^{A\prec B}_{AB_I}\otimes \ketbra{0}{0}_{C_I} \!+ \widetilde{W}^{A\prec B}_{AB_I}\otimes \ketbra{1}{1}_{C_I} \\
    &+ F\otimes \ketbra{1}{0}_{C_I} + F^\dagger\otimes \ketbra{0}{1}_{C_I})\, ,
\end{split}
\end{gather}
where, for simplicity, we omitted the respective identity matrices. Now, choosing $W^{A\prec B}_{AB_I} = \tfrac{1}{4} \ident_{AB} + \alpha\sigma_{A_I}^x\sigma_{A_O}^x\sigma_{B_I}^x$ and $\widetilde{W}^{A\prec B}_{AB_I} = \tfrac{1}{4} \ident_{AB} - \alpha\sigma_{A_I}^x\sigma_{A_O}^x\sigma_{B_I}^x$, we see that both of them are -- for $\alpha \in \mathbbm{R}$ sufficiently small -- proper process matrices with causal order $A\prec B$ (and, importantly, they are not of order $A\|B$). Consequently, for both outcomes $1$ and $0$ one obtains two (different) processes of ordering $A\prec B$. Overall, i.e., when discarding the degrees of freedom $C_I$, we have $\tfrac{1}{2} (W^{A\prec B}_{AB_I} + \widetilde{W}^{A\prec B}_{AB_I}) = \tfrac{1}{4} \ident_{AB}$, which is a process of ordering $A\|B$.  Importantly, conditioning in the $x$-basis yields the two process matrices
\begin{gather}
\begin{split}
   W^{(\pm)} &= \tr_{C_I}(\Upsilon^{A\|B\prec C_I}_{ABC_I}\ketbra{\pm}{\pm}) \\
   &= \tfrac{1}{8}\ident_{AB} \pm (F + F^\dagger)\, \, 
   \end{split}
\end{gather}
with respective probability $p=\tfrac{1}{2}$ (i.e., $2W^{\pm}$ is a properly normalised process matrix).
Here, we see that, for $F=0$, we cannot obtain process matrices $W^{(\pm)}$ of opposing causal order $B\prec A$. However, by choosing $F$ appropriately, both processes $W^{(\pm)}$ can indeed be of causal order $B\prec A$. This is, for example, achieved by setting $F= \beta \sigma_{A_I}^x\sigma_{B_I}^x\sigma_{B_{O}}^x$, in which case we have 
\begin{gather}
\label{eqn::CausalOrderDeterministic}
     W^{(\pm)} = \tfrac{1}{8}\ident_{AB} \pm \beta \sigma_{A_I}^x\sigma_{B_I}^x\sigma_{B_{O}}^x\, , 
\end{gather}
which, for appropriately chosen $\beta \in \mathbbm{R}$, is positive and satisfies -- up to normalisation -- $W^{(\pm)} = \ident_{A_O}\otimes W_{A_IB}^{(\pm)B\prec A}$, but not $W^{(\pm)} = \ident_{B_O}\otimes W_{AB_I}^{(\pm)A\prec B}$ implying that both of them have causal order $B\prec A$. 

It remains to show that these choices actually lead to a proper comb $\Upsilon_{ABC_I}^{A\|B\prec C_I}$. First, from\eq~\eqref{eqn::CausalOrderDeterministic} we see that that $\Upsilon_{ABC_I}^{A\|B\prec C_I}$ indeed satisfies the relevant causality constraints, as $\tr_{C_I}\Upsilon_{ABC_I}^{A\|B\prec C_I} = \tfrac{1}{4}\ident_{AB}$. On the other hand, with the choices we made, the smallest eigenvalue of $ \Upsilon_{ABC_I}^{A\|B\prec C_I}$ is given by $\tfrac{1}{8}(1-4\sqrt{\alpha^2 +\beta^2})$, which can be satisfied by choosing $|\alpha|\neq 0$ and $|\beta| \neq 0$ sufficiently small.
\end{proof}

While the above $\Upsilon_{ABC_I}^{A\|B\prec C_I}$ yields a different process matrix for each of the considered outcomes, the causal ordering of these processes only depends on the respective instrument, \textit{not} the specific outcome of the instrument; conditioning in the basis $\{\ket{0/1}\}$ leads to processes of order $A\prec B$, while conditioning in the $\{\ket{\pm}\}$ yields processes of ordering $B\prec A$. Consequently, causal order indeed becomes -- in a well-defined sense -- an instrument-dependent property and can be chosen at will by Charlie.

It is worth clarifying that, in the above scheme, Charlie is not predetermining the causal order or signalling to Alice and Bob which causal order he wishes to see. Importantly, Charlie can choose the causal direction \textit{after} the experiment (in Alice's and Bob's laboratories) has already concluded. Therefore, this process is a causal version of the famous \textit{delayed-choice} experiment by Wheeler~\cite{wheeler, peruzzo} that renders the chicken-egg debate fundamentally unresolvable. The instrument-dependence of causality here is reminiscent of tachyons, i.e., particles that travel faster than the speed of light~\cite{tachyons}. In this case a `reinterpretation principle' is put forth as different Lorentz frames will see different causal orders; some will see a particle emitted at $A$ and absorbed at $B$, while others will see the same particle emitted at $B$ and absorbed at $A$. This means that even in absence of the theory of relativity, as in our case, a `reinterpretation principle' may be necessary in quantum mechanics.

\subsection{Causal order and conditioning in a single basis}

While, as we have seen, it is possible to devise a process such that the causal order of the resulting conditioned process matrix can be changed by changing the respective measurement basis, it is \textit{not} possible to devise a process and an instrument such that conditioning on \textit{either} outcome leads to processes of opposite causal order. Specifically, we have the following no-go Observation:
\begin{observation}
\label{obs::nogo}
Conditioning on two different outcomes of a fixed measurement basis cannot yield two causally definite process matrices of opposite causal orders. 
\end{observation}
This Observation mirrors similar results in the unconditional case discussed in Refs. ~\cite{yokojima_consequences_2020, costa_no-go_2020}. There, it was shown that, in many simple cases, it is not possible to directly -- i.e., without an additional flag system -- superpose processes of opposing causal order. Here, on the other hand, we show that, under the assumption that the overall process is causally ordered, it is not even possible to obtain processes of opposing causal order when conditioning on an additional system.
\begin{proof}
Let us denote the process matrix obtained when conditioning on outcome $0$ by $W$, and the one obtained when conditioning on outcome $1$ by $W'$. Assuming that the process used for conditioning was of the causal order $A\prec B \prec C_I$ (the other case follows in the same vein), we have 
\begin{gather}
    qW + (1-q)W' = \Gamma^{A\prec B}\, ,
\end{gather}
where $q$ is the probability to observe outcome $0$ and the overall process matrix with definite causal order $A\prec B$ is -- to distinguish it from the conditioned ones -- denoted by $\Gamma^{A\prec B}$. Consequently, $\Gamma^{A\prec B}$ is of the form $\Gamma^{A\prec B} = \ident_{B_O} \otimes \Gamma^{A\prec B}_{AB_I}$. Now, assuming that $W$ and $W'$ are of opposite causal orders $A\prec B$ and $B\prec A$, respectively, we see that 
\begin{gather}
\label{eqn::CausalDecomp}
\begin{split}
 \ident_{B_O} \otimes \Gamma^{A\prec B}_{AB_I}
 =& q \ident_{B_O} \otimes W^{A\prec B}_{AB_I}\\
 &+ (1-q)\ident_{A_O} \otimes W^{\prime B\prec A}_{BA_I} \, .
 \end{split}
 \end{gather}
Since $W^{A\prec B}_{AB_I}$ is Hermitian, it can be decomposed in terms of generalized Pauli matrices, i.e., $W^{A\prec B}_{AB_I} = \sum_{ijkl} c_{ijk\ell} \sigma_{A_I}^i\otimes \sigma_{A_O}^j\otimes \sigma_{B_I}^k\otimes \sigma_{B_O}^\ell$. If this decomposition contains any term that has a non-trivial (i.e., $\neq \ident_{B_O}$) generalized Pauli matrix on $B_O$, then\eq~\eqref{eqn::CausalDecomp} cannot hold. Consequently, $W^{\prime B\prec A}_{BA_I}$ is of the form $W^{\prime B\prec A}_{BA_I} = \ident_{B_O} \otimes \rho_{A_IB_I}$, implying that $W^{\prime B\prec A}$ is of the form $A\|B$, which is not of opposite causal order than $W^{A\prec B}$.
\end{proof}
Importantly, the above Observation is independent of the details of the causal circuit employed, and only relies on the requirement that $qW^{(0)} + (1-q)W^{(1)}$ must be causally ordered (or of the form $A\|B$). We emphasize though, that this reasoning only holds for conditioning with two outcomes; for three possible outcomes, it is straightforward to construct cases where, for example, the resulting $W^{(0)}$ is causally ordered $A\prec B$, while $W^{(1)}$ and $W^{(2)}$ are causally ordered $B\prec A$. This even holds true for purely classical processes, i.e., cases where all involved process matrices are diagonal in the same product basis. 

To see this, consider an arbitrary process matrix $W^{(0)}$ with causal ordering $A\prec B$ that is diagonal in the basis $\{\ket{i_{A_I}j_{A_O}k_{B_I}\ell_{B_O}}\}$, where $\ket{m_{X}}$ denotes an element of the computational basis of $\Hcal_X$. Now, choosing a (classical) process matrix $W^{(1)} = \ident_{A_O} \otimes  D_{A_IB_O} \otimes \rho_{B_I}$ with causal ordering $B\prec A$, where $D_{A_IB_O} = \sum_\ell \ketbra{\ell}{\ell}_{A_I} \otimes \ketbra{\ell}{\ell}_{B_O}$ is the Choi state of the completely dephasing map, and $\rho_{B_I}$ is an arbitrary state that is diagonal in the basis $\{\ket{k_{B_I}}\}$, we can find an appropriate $W^{(2)}$. As mentioned below the proof of Obs.~\ref{obs::OppCausOrd}, there always exists a $p \in (0,1]$ such that 
\begin{gather}
   \frac{\ident_{AB}}{d_{A_IB_I}} \geq pW^{(1)} = p \ \ident_{A_O} \otimes  D_{A_IB_O} \otimes \rho_{B_I}\, .
\end{gather}
Thus, $W^{(2)} := \tfrac{1}{1-p}(\tfrac{\ident_{AB}}{d_{A_OB_O}} - pW^{(1)})$, is a proper process matrix (with causal order $B\prec A$) and we see that 
\begin{gather}
\begin{split}
    \Upsilon_{ABC_I} \! =& q W^{(0)}\! \otimes\! \ketbra{0}{0}_{C_I} \!+\! (1-q) p W^{(1)} \!\otimes\! \ketbra{1}{1}_{C_I}\\ 
    &+ (1-q) (1-p) W^{(2)} \otimes \ketbra{2}{2}_{C_I}
\end{split}
\end{gather}
is a properly causally ordered comb (with order $A\prec B\prec C$), as it is positive and satisfies
\begin{gather}
    \tr_{C_I}\Upsilon_{ABC_I} = qW^{(0)} + (1-q)\frac{\ident_{AB}}{d_{A_IB_I}}\, .
\end{gather}
Conditioning the process $\Upsilon_{ABC_I}$ on outcome $0$ when measuring $C_I$ the yields $W^{(0)}$, which, by assumption is of causal order $A\prec B$, while conditioning on $1$ yields $W^{(1)}$, which, by construction, is of causal order $B\prec A$. Finally, conditioning on outcome $2$ yields the process matrix $W^{(2)}$, which is also of causal order $B\prec A$.

Allowing for more than two outcomes also admits a direct connection to Obs.~\ref{obs::OppCausOrd}, as it enables one to mimic measurements in two different bases by means of one single instrument. For example, choosing a generalized measurement with corresponding POVM elements 
\begin{gather}
    \begin{split}
& E^{(0)} = \tfrac{\sqrt{2}}{1+\sqrt{2}} \ketbra{0}{0}_{C_I}, \quad E^{(1)} = \tfrac{\sqrt{2}}{1+\sqrt{2}} \ketbra{+}{+}_{C_I}, \\
&E^{(2)} = \ident_{C_I} - E^{(0)} - E^{(2)},
    \end{split}
\end{gather}
it is possible to condition on both $\ket{0}_{C_I}$ and $\ket{+}_{C_I}$ with a single measurement setting -- as considered in the proof of Obs.~\ref{obs::OppCausOrd}. This, then, possibly leads to conditioned processes with opposing causal order, with the caveat that there is an additional third outcome, which, as long as $\ket{0}_{C_I}$ and $\ket{+}_{C_I}$ yield proper process matrices, corresponds to a proper process matrix as well. Additionally, similar to the discussion below Obs.~\ref{obs::OppCausOrd}, this realisation of opposing causal orders is inherently probabilistic, as there is always one additional (third) outcome that leads to a process of indefinite causal order.

Besides only applying to two outcomes, the reasoning that led to Obs.~\ref{obs::nogo} necessarily \textit{only} holds if the employed circuit has a definite causal order; here, the difference between the quantum switch and our procedure becomes apparent once more; discarding the control qubit of the quantum switch leaves the remaining degrees of freedom in a convex mixture of opposing causal orders. This is in contrast to the above reasoning, where we employed the fact that tracing out Charlie's degrees of freedom yields a causally ordered process whenever the underlying process $\Upsilon_{ABC_I}$ is causally ordered. Consequently, using a quantum switch allows one to condition onto two opposing causal orders by means of one basis -- and two outcomes -- only, a feat not possible for causally ordered $\Upsilon_{ABC_I}$.
\section{Conclusions and Outlook}

The exotic nature and theoretic appeal of causally indefinite processes is undeniable. However, their foundational and practical importance is still under debate. Here, by focusing on physically realisable processes, we have elucidated the ontological status of implementation schemes of causal indefiniteness by connecting them to causally ordered processes without non-classical correlations (in the relevant splitting) via a conditioning scheme. In addition, we have constructed causally ordered tri-partite processes that lead to a causally indefinite process for \textit{any} conditioning of the third party. Finally, building upon these methods we have demonstrated an analogue of the delayed-choice (thought) experiment for causal orders. Our results add to the growing body of work that underlines the foundational importance of causally indefinite process matrices, and they show that the list of exotic quantum phenomena is yet to be fully mapped out.

Our work highlights striking basis-dependent and basis-independent features of causality in quantum mechanics. Concretely, we have shown that causal order can be basis-dependent (in a precise sense): Conditioning in two different bases can lead to process matrices that have opposing causal orders. Importantly, this basis-dependence of causal ordering can be implemented deterministically, such that the choice of conditioning instrument also allows for choosing the observed causal order. This unresolvability of the chicken-and-egg dilemma in quantum mechanics~\cite{choaug_17_quantum_2018, *ChickEgg} has been studied in the context of the quantum switch. There, however, the process itself is not causally separable and the reduced process, i.e., the process when the degrees of freedom $C_I$ are discarded corresponds to a convex mixture of opposing causal orders. Here, we have demonstrated here that there are cases where this chicken-and-egg dilemma cannot be decided even under the assumption of global causal order. We showed that this, however, can only occur if genuine quantum correlations between the relevant degrees of freedom and the conditioning degrees of freedom are present in the conditioning comb. Naturally, such an effect is not at odds with special relativity, as it only holds in a conditioning sense, but not if the respective degrees of freedom of $C_I$ are discarded.

This phenomenon can be thought of as a variant of the delayed-choice experiment and warrants an analysis in the device independent setting~\cite{PhysRevLett.120.190401}. These results also complement those of~\cite{pienaar_time-reversible}, where the effect of a restriction of the possible instruments on the perceived causal order was studied. On the other hand, in contrast to, for example, the quantum switch, it is not possible to use a causally ordered comb to condition onto two opposing causal orders by means of only one instrument with two outcomes. Our results thus complement similar findings for the unconditional case~\cite{yokojima_consequences_2020,costa_no-go_2020}.

Furthermore, we analysed the `robustness' of causally non-separable process matrices with respect to the choice of conditioning basis. Specifically, we showed that adding entanglement between $AB$ and $C_I$, or, equivalently, adding coherent control over the conditioned process matrices, while still keeping the resulting comb properly causally ordered, and ensuring that all conditioning leads to proper process matrices, can lead to scenarios where conditioning in \textit{any} basis yields a causally non-separable process matrix. In addition, the explicit example we provided displayed some resistance against noise in the conditioning process, making it, in principle, amenable to experimental testing.

While for the deterministic implementation of opposite causal orders, entanglement in the splitting $AB:C_I$ is a necessary prerequisite, it is not \textit{a priori} clear if this is also the case for the stability advantage in the realisation of causal non-separability; in our analysis, all the causally ordered processes that yielded a stability advantage over the classically correlated case in\eq~\eqref{eqn::WocbIncoherent} were entangled in the splitting $AB:C_I$, but it is unclear if entanglement is indeed responsible for this advantage; in principle, there could be separable causally ordered processes that yield causally non-separable process matrices $W(q,\vartheta)$ in \textit{any} conditioning basis. However, we conjecture that there is, again, an interconversion of properties, and entanglement is necessary for full stability with respect to the conditioning basis.

Lastly, it is as of yet unclear how generic the property of full stability is with respect to the choice of measurement basis. Answering this question is hindered by the fact that a randomly chosen causally ordered comb $\Upsilon_{ABC_I}$ does not generally yield a proper process matrix on $AB$ when conditioned on measurements on $C_I$. More precisely, as \textit{any} positive matrix $M$ on $AB$ can be `realised'~\footnote{As mentioned, for matrices that are not proper processes, the conditioning probability depends on the employed instruments and the conditioning procedure is somewhat}ill-defined. by means of a causally ordered $\Upsilon_{ABC_I}$, the probability to realise proper process matrices is vanishing for a randomly chosen $\Upsilon_{ABC_I}$. Consequently, results with respect to the prevalence of fully stable combs have to be deferred to future work.

 \vspace{10pt}
\acknowledgments{We thank Jessica Bavaresco and Jacques Pienaar for valuable discussions, and Johanna Sch\"afer for illustratorial assistance. SM acknowledges funding from the Austrian  Science  Fund  (FWF):  ZK3  (Zukunftkolleg) and  Y879-N27  (START  project), the  European  Union's Horizon 2020 research and innovation programme under the Marie Sk{\l}odowska Curie grant agreement No 801110, and the Austrian Federal Ministry of Education, Science and Research (BMBWF). KM is supported through Australian Research Council Future Fellowship FT160100073.}
\FloatBarrier
\bibliographystyle{prx_bib}
\bibliography{references.bib}
 
\section*{Appendices} 
\appendix



\section{SDP for causal robustness}
\label{app:CausRob}

Here, we provide the SDP for the computation of the causal robustness, that is used throughout the paper. To this end, we first express the definition of causal robustness (Eq.~\eqref{eqn::Robustness}) as
\begin{flushleft}
\begingroup
\renewcommand{\arraystretch}{1.25}
\begin{tabular}{ l l } 
 \textbf{minimize:} & $s$ \\ 
 \textbf{subject to:} & $\tfrac{W+ sW'}{1+s} = pW^{A\prec B} + (1-q)W^{B\prec A}$,\\  
  &$W^{A\prec B}  = {}_{B_O}W^{A\prec B}$, \\
  &  ${}_{B_OA_O} W^{A\prec B} = {}_{B_OB_IA_O} W^{A\prec B}$,\\
  &  $ W^{B\prec A} = {}_{A_O} W^{B\prec A}$, \\
  & ${}_{A_OB_O} W^{B\prec A} = {}_{A_OA_IB_O} W^{B\prec A}$,\\
  & $L_V(W') = W'$,\\
  & $s, W^{A\prec B},  W^{B\prec A}, W' \geq 0$, $p \in [0,1]$ \\
  &$\tr W^{A\prec B} =  \tr W^{B\prec A} = \tr W' = d_{A_O}d_{B_0}$\, , 
\end{tabular}
\endgroup
\end{flushleft}
where we have introduced the projector 
\begin{gather}
\begin{split}
    L_V(W) &= {}_{A_O}W + {}_{B_O}W - {}_{A_OB_O}W - {}_{B_IB_O}W \\
    &\phantom{=} + {}_{A_OB_IB_O}W - {}_{A_IA_O}W + {}_{A_IA_OB_O}W\, ,
\end{split}    
\end{gather}
and the operators ${}_XW = \tfrac{1}{d_X} \ident_X \otimes \tr_X(W)$. The requirements of the above program on $W^{A\prec B}$ and $W^{B\prec A}$ ensure that they are causally ordered -- i.e., satisfy Eqs.~\eqref{eqn::Causality_constraints} -- while the requirements on $W'$ ensure that it is a proper process matrix -- i.e, satisfies\eq~\eqref{eqn::Def_W} (see Ref.~\cite{araujo_witnessing_2015} for more details). In the form presented above, this program is not yet an SDP, but can be straightforwardly rewritten into one.

Setting $\widetilde{W}^{A\prec B} = (1+s)pW^{A\prec B}$, $\widetilde{W}^{B\prec A} = (1+s)(1-p)W^{B\prec A}$ and using $sW'\geq 0$, the first line of the above program can be rewritten as 
\begin{gather}
    \widetilde W^{A\prec B} + \widetilde W^{B\prec A} - W \geq 0\, .
\end{gather}
With this, $\Ccal_R(W)$ can then be obtained as the solution of the SDP
\begin{flushleft}
\begingroup
\renewcommand{\arraystretch}{1.25}
\begin{tabular}{ l l } 
 \textbf{minimize:} &$\tfrac{1}{d_{A_O}d_{B_O}}\tr(\widetilde W^{A\prec B} + \widetilde W^{B\prec A})- 1$ \\ 
 \textbf{subject to:} & $\widetilde W^{A\prec B} + \widetilde W^{B\prec A} - W \geq 0$,\\  
  &$\widetilde{W}^{A\prec B}  = {}_{B_O}\widetilde{W}^{A\prec B}$, \\
  &  ${}_{B_OA_O} \widetilde W^{A\prec B} = {}_{B_OB_IA_O} \widetilde W^{A\prec B}$,\\
  &  $ \widetilde W^{B\prec A} = {}_{A_O} \widetilde W^{B\prec A}$, \\
  & ${}_{A_OB_O} \widetilde W^{B\prec A} = {}_{A_OA_IB_O} \widetilde W^{B\prec A}$,\\
  & $\widetilde W^{A\prec B}, \widetilde W^{B\prec A} \geq 0$\, ,  \\
\end{tabular}
\endgroup
\end{flushleft}
which is the SDP used throughout for the computation of $\Ccal_R(W)$.
\section{Valid \texorpdfstring{$F$}{}-terms in \texorpdfstring{$\Upsilon^F_{ABC_I}$}{}}
\label{app::Fterms}

Here, we derive the requirements on the $F$ terms in 
\begin{gather}
 \label{eqn::WF_Appendix}
    \begin{split}
 \Upsilon^F_{ABC_I} &= qW^{(\mathrm{OCB})} + (1-q)W^\# \\
  &\phantom{=}+ 2\sqrt{q(1-q)} (e^{-\iu \vartheta} F + e^{\iu \vartheta} F^\dagger)   
  \end{split}
\end{gather}
mentioned in the main text. Specifically, there are two conditions on $\Upsilon^F_{ABC_I}$ -- leading to the corresponding requirements for $F$ that need to be fulfilled. First, $\Upsilon^F_{ABC_I}$ must be positive, so that it is a proper causally ordered process (the causality constraints are satisfied by construction). Second, all conditioned process matrices obtained from $\Upsilon^F_{ABC_I}$ must be proper process matrices, i.e., they must satisfy\eq~\eqref{eqn::Def_W}. We start with positivity. 

\subsection{Positivity of \texorpdfstring{$\Upsilon^F_{ABC_I}$}{}}
\label{app::Positivity}
Using the eigendecompositions for $W^{(\mathrm{OCB})} = \tfrac{1}{2} \sum_{i = 1}^8 \ketbra{\Psi_i}{\Psi_i}$ and $W^\# = \tfrac{1}{2} \sum_{j = 1}^8 \ketbra{\Psi_j^\perp}{\Psi_j^\perp}$,\eq~\eqref{eqn::WF_Appendix} reads 
\begin{align}
\notag
 &\Upsilon^F_{ABC_I} \\
 \notag
     &= \sum_{i,j=1}^8(\tfrac{1}{2} \ketbra{\Psi_i}{\Psi_i} \otimes \ketbra{0}{0}_{C_I}  + \tfrac{1}{2}\ketbra{\Psi_j^\perp}{\Psi_j^\perp} \otimes \ketbra{1}{1}_{C_I} \\
\label{eqn::WF_app}
     &\phantom{=}+ F\otimes \ketbra{0}{1}_{C_I} + F^\dagger \otimes \ketbra{1}{0}_{C_I})\, .
\end{align}
Now, projection on a vector $\ket{\Psi_k}$ yields
\begin{gather}
\begin{split}
    &\braket{\Psi_k|\Upsilon^F_{ABC_I}|\Psi_k} \\
    &= \tfrac{1}{2} \ketbra{0}{0}_{C_I} + f_k \ketbra{0}{1}_{C_I} + f_k^\ast \ketbra{1}{0}_{C_I}\, ,
    \end{split}
\end{gather}
where $f_k = \braket{\Psi_k|F|\Psi_k} $. In matrix form, the above equation reads 
\begin{gather}
\label{eqn::cd_terms_app}
    \braket{\Psi_k|\Upsilon^F_{ABC_I}|\Psi_k} = \left( \begin{array}{cc} 
    \tfrac{1}{2} & f_k \\ f_k^\ast & 0 
    \end{array} \right)\, ,
\end{gather}
which has eigenvalues $\lambda_{\pm} = \tfrac{1}{2} (\frac{1}{2} \pm \sqrt{\tfrac{1}{4} + 4|f_k|^2})$. For $\Upsilon^F_{ABC_I}$ to be positive, we thus require that $f_k = 0$ for all $k\in \{1,\hdots,8\}$. Running the same  argument for the eigenvectors of $W^\#$ shows that $F$ cannot contain any terms of the form $\ketbra{\Psi_j^\perp}{\Psi_j^\perp}$, implying that it is of the form 
\begin{gather}
F = \sum_{i,j=1}^8 (c_{ij} \ketbra{\Psi_i}{\Psi_j^\perp} + d_{ij} \ketbra{\Psi_i^\perp}{\Psi_j}) \, ,
\end{gather}
with $c_{ij}, d_{ij} \in \mathbbm{C}$. This also implies $\tr F = 0$, as mentioned in the main text. Furthermore, we can show that $d_{ij} = 0$ is necessary for $\Upsilon^F_{ABC_I}$ to be positive. To this end, we insert\eq~\eqref{eqn::cd_terms_app} into\eq~\eqref{eqn::WF_app}, which yields

\begin{align}
\notag
     &\Upsilon^F_{ABC_I} \\
     \notag
     &= \sum_{i,j=1}^8(\tfrac{1}{2} \ketbra{\Psi_i}{\Psi_i} \otimes \ketbra{0}{0}_{C_I}  + \tfrac{1}{2} \ketbra{\Psi_j^\perp}{\Psi_j^\perp} \otimes \ketbra{1}{1}_{C_I}) \\
     \notag
     &\phantom{=} +\sum_{i,j=1}^8(c_{ij} \ketbra{\Psi_i}{\Psi_j^\perp}  + d_{ij}\ketbra{\Psi_i^\perp }{\Psi_j})  \otimes \ketbra{0}{1}_{C_I} \\
     \label{eqn::WF_coeff_app}
     &\phantom{=}+ \sum_{i,j=1}^8 (c_{ij}^\ast \ketbra{\Psi_j^\perp}{\Psi_i} + d^\ast_{ij} \ketbra{\Psi_j}{\Psi_i^\perp}) \otimes \ketbra{1}{0}_{C_I} \, .
\end{align}
Now, collecting the terms with coefficients $d_{ij}$, we set 
\begin{gather}
    \begin{split}
    D :=  \sum_{i,j=1}^8&(d_{ij} \ketbra{\Psi_i^\perp}{\Psi_j} \otimes \ketbra{0}{1}_{C_I} \\
    &\phantom{asdf}+  d_{ij}^\ast \ketbra{\Psi_j}{\Psi_i^\perp} \otimes \ketbra{1}{0}_{C_I})\, , 
\end{split}
\end{gather}
and denote the remaining terms by $G$, such that $\Upsilon^F_{ABC_I} = D + G$. We have $D\cdot G = 0$ and $\tr(D) = 0$. As $D$ is Hermitian, it has real eigenvalues, and as $\tr(D) = 0$, at least one of these eigenvalues is negative (unless  $D = 0$). Consequently, since the supports of $D$ and $G$ are orthogonal, $D + G$ has at least one negative eigenvalue if $D \neq 0$, in which case $\Upsilon^F_{ABC_I} \ngeq 0$, which contradicts our initial requirement. This implies that all $d_{ij}$ vanish when $\Upsilon^F_{ABC_I}$ is positive. Note that an analogous reasoning does \textit{not} hold for the coefficients $c_{ij}$. Denoting the terms in\eq~\eqref{eqn::WF_coeff_app} that contain the coefficients $c_{ij}$ by $H$, and the remaining ones by $K$ (such that $\Upsilon^F_{ABC_I} = H + K$), it is easy to see that the supports of $H$ and $K$ are not necessarily orthogonal, and the above reasoning for $\{d_{ij}\}$ would not carry over to $\{c_{ij}\}$.

We will return to the explicit positivity conditions when imposing that $\Upsilon^F_{ABC_I}$ is a proper process matrix below, after first further reducing the number of non-vanishing parameters $\{c_{ij}\}$. 

\subsection{\texorpdfstring{$F$}{}-terms and valid conditioned process matrices}
\label{app::validF}
In principle, conditioning allows for the realisation of any type of `process', valid (i.e., satisfying\eq~\eqref{eqn::Def_W}) or not. Naturally, here, we demand that conditioning leads to a proper process matrix, independent of the conditioning basis. While the linear requirements (besides positivity) on a matrix $W$ to be a proper process matrix can be phrased in a basis independent way~\cite{araujo_witnessing_2015}, we choose the characterization in terms of Pauli matrices provided in Ref.~\cite{OreshkovETAL2012}. Specifically, since a process matrix is Hermitian (and, in our case, defined on a four-qubit Hilbert space), it can be decomposed in terms of a Pauli basis as 
\begin{gather}
    W = \sum_{\alpha, \beta, \gamma,\mu=0}^3 w_{\alpha \beta \gamma\mu} \sigma^\alpha_{A_I} \otimes \sigma_{A_O}^\beta \otimes \sigma_{B_I}^\gamma \otimes \sigma_{B_O}^\mu\, ,
\end{gather}
where $\sigma^0_{X} = \ident_X$,  $\sigma^1_{X} = \sigma^x_X$, $\sigma^2_{X} = \sigma^y_X$, and $\sigma^3_{X} = \sigma^z_X$. Due to normalization, we have  $w_{0000} = \tfrac{1}{d_{A_I}d_{A_0}}$. Now, in order for $W$ to be a proper process matrix, it has to be positive, and certain terms in the above decomposition cannot be present. In particular, denoting the respective terms by the Hilbert spaces on which they have non-trivial Pauli matrix (e.g., a term of the form  $\sigma^x_{A_I} \otimes \ident_{A_O} \otimes \sigma_{B_I}^z \otimes \sigma_{B_O}^y$ would be an $A_IB_IB_O$ term), it has been shown~\cite{OreshkovETAL2012} that terms of the form 
\begin{gather}
    \begin{split}
    \Sigma_{NA} = \{&A_O,B_O, A_OB_O, A_IA_O, B_IB_O, A_IA_OB_O, \\
    &A_OB_IB_O, A_IA_OB_IB_O \}
\end{split}
\end{gather}
are not allowed in the decomposition of $W$. 

As both $W^{(\mathrm{OCB})}$ and $W^\#$ do not contain any terms that are not allowed, neither can $F$, which we denote by the shorthand $\tr(F\sigma^{\Gamma}_{NA})=0$ for all $\sigma^{\Gamma}_{NA} \in \Sigma_{NA}$. It is easy to see that the index $\Gamma$ runs from $1$ to $168$, i.e., there are altogether $168$ Pauli terms that cannot appear in a proper process matrix (defined on a four qubit Hilbert space). With this, we can derive the conditions the parameters $c_{ij}$ have to satisfy for the conditioned process matrices to be proper ones. In particular, setting 
\begin{gather}
    r_{ij}^\Gamma = \tr(\sigma_{NA}^{\Gamma} \ketbra{\Psi_i}{\Psi_j^\perp})\, ,
\end{gather}
we see that the requirement that no Pauli term that is not allowed appears in the decomposition of $F$ leads to 
\begin{gather}
\label{eqn::LinEq_app}
    \sum_{i,j=1}^8 c_{ij} r_{ij}^\Gamma = 0 \quad \forall \ \Gamma\in \{1,\hdots, 168\}\, .
\end{gather}
This linear equation can be readily solved to determine the coefficients $\{c_{ij}\}$. To avoid ambiguity, we explicitly provide the eigenvectors of $W^\mathrm{(OCB)}$ and $W^\#$ as well as the ordering we choose: 
\begin{align}
&\ket{\Psi_1} = \tfrac{1}{4-2\sqrt{2}}[(\sqrt{2}-1)\ket{1101} + \ket{1111}]\, ,\\
&\ket{\Psi_2} =  \tfrac{1}{4-2\sqrt{2}}[(1-\sqrt{2})\ket{1100} + \ket{1110}]\, , \\
&\ket{\Psi_3} = \tfrac{1}{4+2\sqrt{2}}[(1 + \sqrt{2})\ket{1001} +\ket{1011})]\, ,\\
&\ket{\Psi_4} =  \tfrac{1}{4+2\sqrt{2}}[\ket{1010}-(1 + \sqrt{2}) \ket{1000}]\, ,\\
&\ket{\Psi_5} = \tfrac{1}{4-2\sqrt{2}}[(1 - \sqrt{2})\ket{0101} + \ket{0111}]\, ,\\
&\ket{\Psi_6} = \tfrac{1}{4-2\sqrt{2}}[(\sqrt{2}-1)\ket{0100} + \ket{0110}\, ,\\
&\ket{\Psi_7} = \tfrac{1}{4+2\sqrt{2}} [\ket{0011} -(1+\sqrt{2}) \ket{0001}]\, ,\\
&\ket{\Psi_8} = \tfrac{1}{4+2\sqrt{2}}[(1 + \sqrt{2})\ket{0000} + \ket{0010}]\, ,\\
\notag \\
&\ket{\Psi_1^\perp} = \tfrac{1}{4-2\sqrt{2}}[( \sqrt{2}-1)\ket{1101}+ \ket{1111}]\, ,\\
&\ket{\Psi_2^\perp} = \tfrac{1}{4+2\sqrt{2}}[(1 +\sqrt{2})\ket{1100} + \ket{1110}]\, , \\
&\ket{\Psi_3^\perp} = \tfrac{1}{4-2\sqrt{2}}[(1 - \sqrt{2}) \ket{1001} + \ket{1011}]\, ,\\
&\ket{\Psi_4^\perp} = \tfrac{1}{4-2\sqrt{2}}[( \sqrt{2}-1)\ket{1000} +\ket{1010}]\, ,\\
&\ket{\Psi_5^\perp} = \tfrac{1}{4+2\sqrt{2}}[(1 + \sqrt{2})\ket{0101} +\ket{0111}]\, ,\\
&\ket{\Psi_6^\perp} = \tfrac{1}{4+2\sqrt{2}}[\ket{0110} - (1 + \sqrt{2})\ket{0100}]\, ,\\
&\ket{\Psi_7^\perp} = \tfrac{1}{4-2\sqrt{2}}[(\sqrt{2}-1)\ket{0001} + \ket{0011}]\, ,\\
&\ket{\Psi_8^\perp} = \tfrac{1}{4-2\sqrt{2}}[(1 - \sqrt{2})\ket{0000} + \ket{0010}]\, .
\end{align}
With this ordering in mind, solving\eq~\eqref{eqn::LinEq_app} yields three free parameters -- we choose  $\{c_{11},c_{15}, c_{51}\}$ -- and
\begin{gather}
\label{eqn::Coefficients_app}
\begin{array}{llll}
c_{22} = -c_{11},\ &c_{26} = -c_{15}, \
&c_{37} = c_{15},\ &c_{44} = -c_{11}, \\
c_{48} = -c_{15},\ &c_{55} = -c_{11},\
&c_{62} = -c_{51},\ &c_{66} = c_{11},\\  
c_{73} = c_{51},\ &c_{77} = -c_{11},\ 
&c_{84} = -c_{51},\ &c_{88} = c_{11}\, ,
\end{array}
\end{gather}
while all other coefficients vanish. Each choice of coefficients $\{c_{11},c_{15}, c_{51}\}$ then provides a proper conditioned process matrix independent of the basis with respect to which conditioning takes place, as long as the remaining coefficients are computed according to Eqs.~\eqref{eqn::Coefficients_app}, and the corresponding $\Upsilon^F_{ABC_I}$ is positive.
 
\subsection{Positivity of \texorpdfstring{$\Upsilon^F_{ABC_I}$}{} revisited}
\label{app::Pos_Rev}
Having reduced the number of non-vanishing coefficients $\{c_{ij}\}$, we can now find the explicit ranges, for which they lead to positive (and thus valid) process matrices $\Upsilon^F_{ABC_I}$. Inserting the conditions~\eqref{eqn::Coefficients_app} into the definition~\eqref{eqn::WF_Appendix} of $\Upsilon^F_{ABC_I}$, we can compute the eigenvalues of $\Upsilon^F_{ABC_I}$ with respect to $\{c_{11},c_{15},c_{51}\}$. The smallest of these eigenvalues reads 
\begin{gather}
    \lambda_{\mathrm{min}} = \tfrac{1}{4} - \tfrac{1}{\sqrt{2}} \sqrt{N + \sqrt{N^2 - 4|P|^2}}\,,
    \end{gather}
 where $N =   2|c_{11}|^2 + |c_{15}|^2 + |c_{51}|^2$ and  $P = c_{11}^2 + c_{15}c_{51}$. Demanding $\lambda_{\mathrm{min}} \geq 0$ then yields the requirement 
 \begin{gather}
 \label{eqn::Nsqrt}
     N + \sqrt{N^2 - 4|P|^2} \leq \tfrac{1}{8}\, .
 \end{gather}
For the special case of $c_{15} = c_{51} = 0$, this implies
\begin{gather}
    |c_{11}| \leq \tfrac{1}{4}\, .
\end{gather}
On the other hand, if $c_{11} = c_{15} = c_{51} =: c$, then\eq~\eqref{eqn::Nsqrt} implies
\begin{gather}
    |c| \leq \tfrac{1}{4\sqrt{2}}\, ,
\end{gather}
as mentioned in the main text. Furthermore, under the assumption $c_{15} = e^{\iu \varphi_1} c_{11}$ and $c_{51} = e^{\iu \varphi_2} c_{11}$, we have 
\begin{gather}
    |c_{11}| \leq \tfrac{1}{\sqrt{32+8\sqrt{8(1-\cos(\varphi_1+\varphi_2))}}}\, .
\end{gather}

For the general case $c_{11} \neq c_{15} \neq c_{51}$,\eq~\eqref{eqn::Nsqrt} yields 
\begin{gather}
\begin{split}
&\tfrac{1}{64}\\
&\geq \tfrac{1}{4}( 2|c_{11}|^2 + |c_{15}|^2 + |c_{51}|^2) - 4|c_{11}^2 + c_{15}c_{51}|^2 \, .
 \end{split}
\end{gather}
Choosing the three parameters $\{c_{11},c_{15},c_{51}\}$ such that they satisfy the above equation (as is done throughout the paper), then yields proper process matrices $\Upsilon^F_{ABC_I}$, and as such proper conditioned process matrices $W(q,\vartheta)$ for all choices of $q$ and $\vartheta$.

\begin{figure}
    \vspace{0.2cm}
    \centering
    \includegraphics[width=0.9\linewidth]{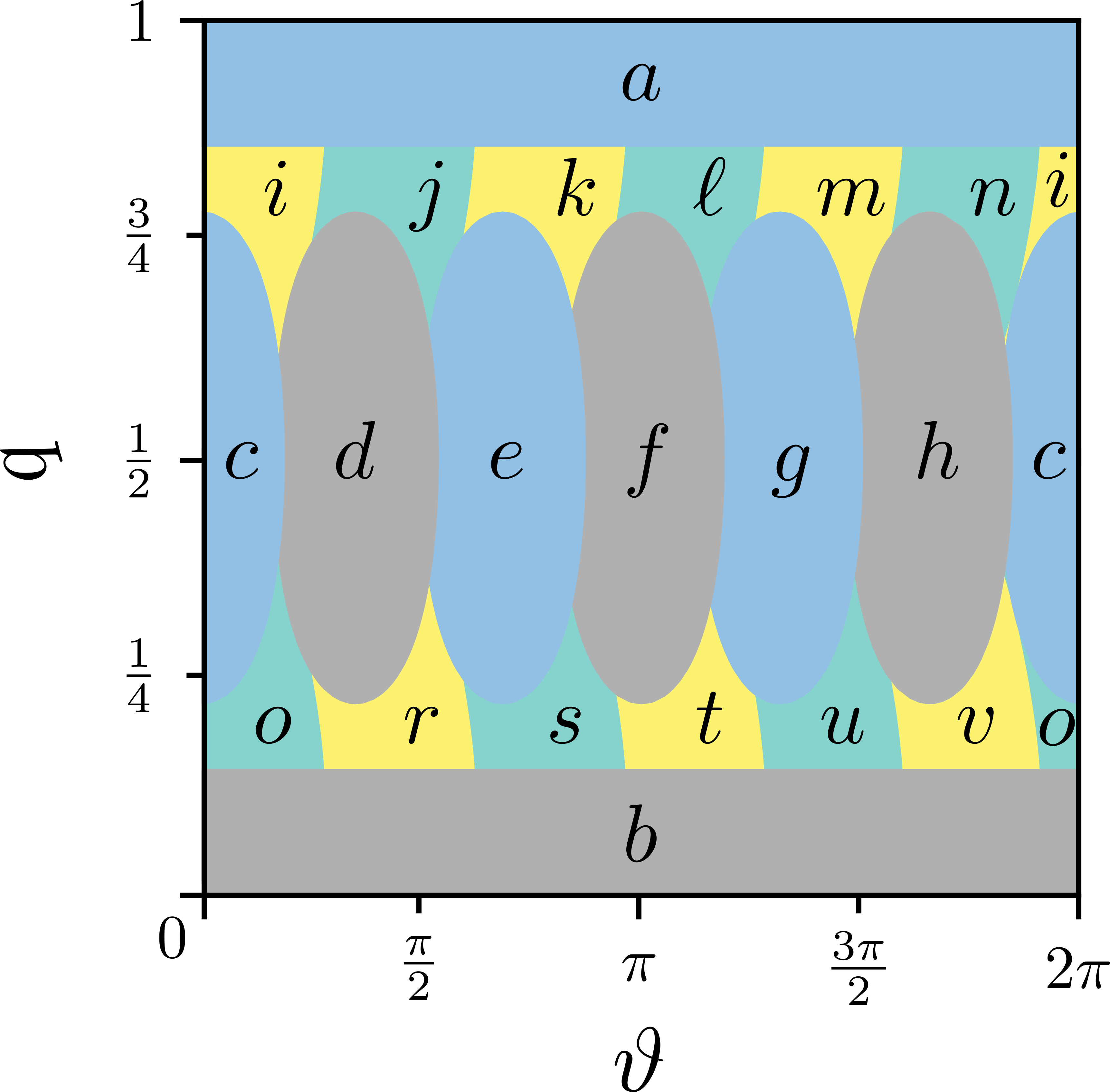}
    \caption{\textbf{Covering the parameter space $(q,\vartheta)$ with witnesses for causal non-separability}. Each of the coloured regions corresponds to an area that is witnessed by a different $S_i$. Depicted are, respectively, not the whole regions the witnesses detect, but only the area necessary to cover the whole parameter space. See Tab.~\ref{tab::Witness_list} for a list of the witnesses each area corresponds to.}
    \label{fig::Covering}
\end{figure}

\begin{table}
\parbox[t]{.3\linewidth}{
\centering
\begin{tabular}[t]{c|c|c}
$A_\alpha$ & $q_\alpha$ & $\vartheta_\alpha$ \\ 
\hline
$a$ & $1$ & $0$ \\ 
\hline
$b$ & $0$ & $0$ \\ 
\hline
$c$ & $0.5$ & $0$ \\ 
\hline
$d$ & $0.5$ & $1.1$ \\ 
\hline
$e$ & $0.5$ & $2.15$\\
\hline
$f$ & $0.5$ & $\pi$ \\ 
\hline
$g$ & $0.5$ & $4.13$ 
\end{tabular}}
\hfill
\parbox[t]{.3\linewidth}{
\centering
\begin{tabular}[t]{c|c|c}
$A_\alpha$ & $q_\alpha$ & $\vartheta_\alpha$ \\ 
\hline
$h$ & $0.5$ & $5.2$ \\ 
\hline
$i$ & $0.74$ & $0$ \\ 
\hline
$j$ & $0.74$ & $1.075$ \\
\hline
$k$ & $0.74$ & $2.15$ \\ 
\hline
$\ell$ &  $0.74$ & $\pi$ \\ 
\hline
$m$ &  $0.74$ & $4.13$ \\  
\hline
$n$ & $0.74$ & $5.11$  
\end{tabular}}
\hfill
\parbox[t]{.3\linewidth}{
\centering
\begin{tabular}[t]{c|c|c}
$A_\alpha$& $q_\alpha$ & $\vartheta_\alpha$ \\
\hline
$o$ & $0.26$ & $0$\\
\hline
$r$ & $0.26$ & $1.075$ \\ 
\hline
$s$ & $0.26$ & $2.15$ \\ 
\hline
$t$ & $0.26$ & $\pi$ \\ 
\hline
$u$ & $0.26$ & $4.13$ \\
\hline
$v$ & $0.26$ & $5.11$
\end{tabular}}
\caption{\textbf{Witnesses used in Fig.}~\ref{fig::Covering}. Each area in Fig.~\ref{fig::Covering} corresponds to the range of parameters for which the causal non-separability of $W(q,\vartheta)$ is detected by the same witness. The employed witnesses $S$ are, respectively, the ideal witnesses for given conditioned process matrices $W(q_*,\vartheta_*)$, i.e., they are proper witnesses and minimize $\tr(SW(q_*,\vartheta_*)$. In the table, the values $(q_*,\vartheta_*)$ which fix the witnesses are listed.}
\label{tab::Witness_list}
\end{table}

\section{Causal non-separability of \texorpdfstring{$W(q,\vartheta)$}{}}
\label{app::Allnonsep}

Here, we show that for the choice $\{c_{11},c_{15},c_{51}\} = \{\tfrac{1}{8},\tfrac{1}{8},\tfrac{1}{8}\}$, \textit{all} resulting process matrices $W(q,\vartheta)$ are causally non-separable. While it is generally hard to analytically compute the causal robustness of a given process matrices, its causal non-separability can be -- just like in the analogous case of entanglement -- determined by means of witnesses~\cite{araujo_witnessing_2015}. These witnesses $S$ are constructed such that if $\tr(SW)<0$, then $W$ is causally non-separable. In~\cite{araujo_witnessing_2015}, it was shown that a witness $S$ of causal non-separability (for two parties) satisfies 
\begin{gather}
\label{eqn::Witness}
S = L_V(S_P) \text{ and } \ident/d_{A_O}d_{B_O} - S = L_V(\Sigma_P),
\end{gather}
where ${}_{A_O}S_P \geq 0,$ ${}_{A_O}S_P \geq 0$, and $\Sigma_P \geq 0$. With this, for any fixed pair $(q_\alpha,\vartheta_\alpha)$ to compute an optimal witness $S_\alpha$ for a conditioned process matrix $W(q_\alpha,\vartheta_\alpha)$  via an SDP~\cite{araujo_witnessing_2015}: 
\begin{flushleft}
\begingroup
\renewcommand{\arraystretch}{1.25}
\begin{tabular}{ l l } 
 \textbf{minimize:} &$\tr(SW(q_\alpha,\vartheta_\alpha))$ \\ 
 \textbf{subject to:} & $S$ is a proper witness of causal \\
 &non-separability (i.e., satisfies\eq~\eqref{eqn::Witness}).
\end{tabular}
\endgroup
\end{flushleft}
Naturally, if a witness $S_\alpha$ detects the causal non-separability of a process matrix $W(q_\alpha,\vartheta_\alpha)$, it can also detect the causal non-separability of process matrices $W(q,\vartheta)$ for parameters $(q,\vartheta)$ in a vicinity of $(q_\alpha,\vartheta_\alpha)$. This allows us to partition the whole parameter space $(q,\vartheta) \in [0,1]\times [0,2\pi]$ into a finite number of areas $A_\alpha$ such that the causal non-separability of each process matrix $W(q,\vartheta)$ with $(q,\vartheta) \in A_\alpha$ is detected by the same witness $S_\alpha$, respectively. To find a sufficient number of witnesses $\{S_\alpha\}$, we simply find the ideal witnesses $\{S_\alpha\}$ for given pairs $(q_\alpha,\vartheta_\alpha)$ by running the above SDP, compute the respective area, in which $\tr(S_\alpha W(q,\vartheta))<0$ holds, until $\bigcup_\alpha A_\alpha$ covers the whole parameter space $(q,\vartheta)$. 

Exemplarily, we explicitly provide the area $A_\alpha$ for three pairs $(q_\alpha,\vartheta_\alpha)$. We start with computing a witness $S_a$ for $W(q_a=1,\vartheta_a=0)$. The corresponding parameter area for which $S_a$ definitely detects causal non-separability is given by $A_a = [0.853553,1] \times [0,2\pi]$. Analogously, the witness $S_b$ for $W(q_b=0,\vartheta_b=0)$ detects causal non-separability for the region $A_0 = [0,0.146447] \times [0,2\pi]$. On the other hand, choosing $(q_f = \tfrac{1}{2},\vartheta_f = \pi)$ as a starting point, the requirement $\tr(S_f W(q,\vartheta))$ for the ideal witness $S_f$ of $W(q_f = \tfrac{1}{2},\vartheta_f = \pi)$ translates to
\begin{gather}
    0.25 + 0.604 \sqrt{(1 - q) q}\cos(\vartheta) < 0\, .
\end{gather}
The corresponding area in which the above inequality is satisfied is depicted in Fig.~\ref{fig::Covering}, where we also provide a complete partitioning of the full parameter space into $20$ areas $A_\alpha$ of parameters $(q,\vartheta)$ that lead to non-separable process matrices that can be detected by the same witness $S_\alpha$. The corresponding values $(q_\alpha,\vartheta_\alpha)$ for which the witnesses $S_\alpha$ are constructed can be found in Tab.~\ref{tab::Witness_list}.

\end{document}